\documentclass[10pt,reqno,notitlepage]{article}

\usepackage{amssymb}
\usepackage{amsthm}
\usepackage{graphicx}
\usepackage{amsmath}
\usepackage{fancyhdr}
\usepackage{natbib}
\usepackage{datetime}
\usepackage[toc]{appendix}
\usepackage[top=1in, bottom=1in, left=1in, right=1in]{geometry}
\usepackage{mathtools}									%
\mathtoolsset{showonlyrefs=true}							%

\newtheorem{theorem}{Theorem}
\newtheorem{lemma}[theorem]{Lemma}								%
\newtheorem{proposition}[theorem]{Proposition}	
	
\newtheorem{assumption}[theorem]{Assumption}	
\theoremstyle{remark}
\newtheorem{remark}{Remark}
\numberwithin{equation}{section}	
\numberwithin{theorem}{section}

\newcommand{\<}{\left\langle} 
\renewcommand{\>}{\right\rangle}
\renewcommand{\(}{\left(}				
\renewcommand{\)}{\right)}
\renewcommand{\[}{\left[}
\renewcommand{\]}{\right]}	
\renewcommand{\vec}[1]{\mathbf{#1}}

\def\E{\mathbb{E}}			
\def\P{\mathbb{P}}										
\def\R{\mathbb{R}}
\def\Rplus{\mathbb{R^{+}}}

\def\A{\mathcal{A}}

\def\L{\mathcal{L}}
\def\B{\mathcal{B}}
\def\D{\mathcal{D}}	
\def\M{\mathcal{M}}		

\def\U{\mathcal{U}}

\def\S{\mathcal{S}}
\def\u{u^{\lam,\eps}}
\def\v{v^{\lam,\eps}}
\def\w{w}
\def\lnt{\L_{\mbox{\tiny{NT}}}}
\def\dmax{\delta_{\mbox{\tiny{max}}}}
\def\sF{{\cal F}}
\def\Sv{{\cal S}}

\def\eps{\varepsilon}

\def\sig{\sigma}

\def\lam{\lambda}

\newcommand{\sigb}{\bar{\sigma}}	%

\def \zhatr {\widehat{\zeta^r}}
\def \zhatl {\widehat{\zeta^l}}
\def \zetar {{\zeta^r}}
\def \zetal {{\zeta^l}}
\def \zetar {{\zeta^r}}
\def \X {{\hat X}}
\def \Y {{\hat Y}}
\def \Z {{\hat Z}}
\def \hatL {{\widehat L}}
\def \hatM {{\widehat M}}

\def\d{\partial}		
\def\dsig{\partial_{\sigma}}	
\def\ind{\mathbb{I}}

\def \l{\ell}
\def \u{u}
\def\lamt{{\eta}}
\def\V{\widehat{V}}
\def \abs#1{| #1 | }
\newcommand{\e}[1]{\operatorname{e}^{#1}}

\def\define{:=}
\newcommand{\pa}{\partial}
\newcommand{\half}{\frac12}

\newcommand{\vega}{\mathcal{V}}
\newcommand{\dpr}{\Delta_{0,\sigma}}

\begin{document}

\title{Optimal Investment with Transaction Costs and Stochastic Volatility}

\author{
Maxim Bichuch
\thanks{
Department of Mathematical Sciences,
Worcester Polytechnic Institute,
Worcester, MA 01609, USA. 
{\tt mbichuch@wpi.edu}. Partially supported by NSF grant DMS-0739195.}
\and Ronnie Sircar
\thanks{
Department of Operations Research
and Financial Engineering,
Princeton University,
Princeton, NJ 08544, USA. 
{\tt sircar@princeton.edu}. Partially supported by NSF grant DMS-1211906.}
}
\date{January 2014, revised \today}
\maketitle

\begin{abstract}
Two major financial market complexities are transaction costs and uncertain volatility, and we analyze their joint impact on the problem of portfolio optimization. When volatility is constant, the transaction costs optimal investment problem has a long history, especially in the use of asymptotic approximations when the cost is small. Under stochastic volatility, but with no transaction costs, the Merton problem under general utility functions can also be analyzed with asymptotic methods. Here, we look at the long-run growth rate problem when both complexities are present, using separation of time scales approximations. This leads to perturbation analysis of an eigenvalue problem. We find the first term in the asymptotic expansion in the time scale parameter, of the optimal long-term growth rate, and of the optimal strategy, for fixed small transaction costs. %
\end{abstract}




\section{Introduction}
The portfolio optimization problem, first analyzed within a continuous time model in \cite{merton69}, ignores two key features that are important for investment decisions, namely transaction costs and uncertain volatility. Both these issues complicate the analysis of the expected utility maximization stochastic control problem, and obtaining closed-form optimal policies, or even numerical approximations, is challenging due to the increase in dimension by incorporating a stochastic volatility variable, and the singular control problem that arises by considering proportional transaction costs. Here, we develop asymptotic approximations for a particular long-run investment goal in a model with transaction costs {\em and} stochastic volatility.

The typical problem has an investor who can invest in a market with one riskless asset (a money market
account), and one risky asset (a stock), and who has to pay a  transaction cost  for selling the stock. The costs are proportional to the dollar amount of the sale, with proportionality constant $\lam>0$.
The investment goal is to maximize the long-term growth rate.
The original works all assumed stocks with constant volatility.
Transaction costs were first introduced into the Merton portfolio problem by
 \cite{MagillConstantinides} and later further investigated by \cite{DumasLuciano}. Their analysis of the infinite time
horizon investment and consumption problem %
gives an insight into the optimal strategy and the existence of a
{\em``no-trade"} (NT) region.
Under certain assumptions, \cite{DavisNorman} provided the first
rigorous analysis of the same infinite time horizon problem. These
assumptions were weakened by \cite{ShreveSoner}, who used viscosity
solutions to also show the smoothness of the value function.

When $\lam>0$, and the volatility is constant, the optimal policy is
to trade as soon as the position is sufficiently far away from the
Merton proportion. More specifically, the agent's optimal policy is to
maintain her position inside a NT region. If the investor's position is initially outside the NT region, she should immediately sell or buy stock in order to move to its boundary. She will trade only when her position is on the boundary of the
NT region, and only as much as necessary to keep it from exiting the
NT region, while no trading occurs in the interior of the region; see
\cite{DavisPanasZariphopoulou}.

There is a trade-off between the amount of transaction costs paid due
to portfolio rebalancing and the width of the NT region. A smaller NT region generally increases the
amount spent paying transaction costs in maintaining the optimal portfolio. Not surprisingly, the same behavior persists when volatility is stochastic, but in this case, the boundaries of NT
region in general will no longer be straight lines as before.
Hence, the approach of this paper, is to find a simple strategy that will be asymptoticaly optimal in both the volatility scaling and transaction costs parameters.  

Small transaction cost asymptotic expansions  (in powers of $\lambda^{1/3}$) were 
used in  \cite{JanecekShreve1} for an infinite horizon investment and consumption problem. This approach allows them to find approximations to the
optimal policy and the optimal long-term growth rate, and is also used in \cite{Bichuch} for a finite horizon optimal investment problem.  The survey article \cite{guasoni} %
describes recent results using so-called {\em shadow price} %
to obtain small transaction cost asymptotics for the optimal investment policy, its
implied welfare, liquidity premium, and trading volume. All of the
above mentioned literature on transaction costs assumes constant
volatility. Some recents exceptions are \cite{KallsenMuhle-Karbe,KallsenMuhle-Karbe1}, where the stock is a general It\^o diffusion, and \cite{SonerTouzi} where the stock is described by a complete (local volatility) model.  We summarize some of this literature and the individual optimization problems and models that they study in Table \ref{tab:T1}.
\begin{table}[htb]
\begin{center}
\scriptsize{\begin{tabular}{| l | c | c | c | c |  }
\hline
{\bf Paper} &  {\bf Model} & {\bf Utility} & {\bf Objective} &  {\bf Solution}	 \\
\hline\hline
\cite{DumasLuciano} & B-S & Power & LTGR & Explicit \\ \hline
\cite{DavisNorman} & B-S & Power & $\infty$-consumption & Numerical\\ \hline
\cite{ShreveSoner} & B-S  & Power & $\infty$-consumption & Viscosity\\ \hline
\cite{DavisPanasZariphopoulou} & B-S & Exponential & Option pricing & Viscosity\\ \hline
\cite{WW} & B-S & Exponential & Option pricing & $\lambda$-expansion\\ \hline
\cite{JanecekShreve1} & B-S & Power & $\infty$-consumption & $\lambda$-expansion\\ \hline
\cite{Bichuch} & B-S & Power & $T<\infty$ & $\lambda$-expansion\\ \hline
\cite{dai2009finite} & B-S & Power & $T<\infty$ & ODEs Free-Bdy\\ \hline
\cite{GMGS} & B-S & Power & LTGR &  $\lambda$-expansion\\ \hline
\cite{GoodmanOstrov} & B-S & General & $T<\infty$ &  $\lambda$-expansion\\ \hline
\cite{ChoiSirbuZitk} & B-S & Power & $\infty$-consumption & ODEs Free-Bdy\\ \hline
\cite{KallsenMuhle-Karbe} & It\^o & General & $T<=\infty$, consumption  & $\lambda$-expansion\\
\hline
\cite{KallsenMuhle-Karbe1} & It\^o &Exponential & Option pricing & $\lambda$-expansion\\ \hline
\cite{SonerTouzi} & Local Vol &General on $\Rplus$ & $T<\infty$ & $\lambda$-expansion\\ \hline
\cite{caflisch} & Stoch vol & Exponential & Option pricing & $\lambda$-SV expansion\\ \hline
This paper & Stoch vol & Power & LTGR & SV expansion\\
\hline
\hline  
\end{tabular}}
\end{center}
\caption{Problems, models and solution approaches. The acronyms used are: B-S = Black-Scholes,  LTGR = Long-Term Growth Rate, SV=Stochastic Volatility, Free-Bdy=Free Boundary.}
\label{tab:T1}
\end{table}

Our approach exploits the fast mean-reversion of volatility (particularly when viewed over a long investment horizon) leading to a singular perturbation analysis of an impulse control problem. We treat the case, when the volatility is slow mean reverting separately. 
This complements multiscale approximations developed for derivatives pricing problems described in \cite{fpss-book} and for optimal hedging and investment problems in \cite{mattias1} and \cite{fsz} respectively. 
Recently, \cite{caflisch} study indifference pricing of European options with exponential utility, fast mean-reverting stochastic volatility and small transaction costs which scale with the volatility time scale.
The current transaction cost problem can be characterized as a free-boundary problem. The 
fast mean-reversion asymptotics for the
finite horizon free boundary problem arising from American options pricing was developed in \cite{american}, and recently there has been interest in similar analysis for perpetual (infinitely-lived) American options (used as part of a real options model) in \cite{ewald}, and for a structural credit risk model in \cite{mcquade}. Here, we also have an infinite horizon free-boundary problem, but it is, in addition, an eigenvalue problem.

In Section \ref{sec:model} of this paper, we introduce our model and
objective function and give the associated Hamilton-Jacobi-Bellman (HJB) equation.
In Section \ref{sec:asymptotic} we perform the asymptotic analysis.
We first consider the fast-scale stochastic volatility in Section \ref{sec:buildSol}, where we find the first correction term in
the power expansion of the value function, and as a result also find
the corresponding term in the power expansion of the optimal
boundary.
We perform similar analysis in the case of slow-scale stochastic volatility in Section \ref{sec:buildSolSlow}. 
In Section \ref{sec:results} we show numerical calculations based on our results, and give an alternative intuitive explanation to the findings.
We summarize the results obtained in the paper in Section \ref{sec:conclusion}, and leave some technical computations to the Appendix.

\section{A Class of Stochastic Volatility Models with Transaction Costs}\label{sec:model}
An investor can allocate capital between two assets -- a risk-free money market account with constant rate of interest $r$, and risky stock $S$ that evolves according to the following stochastic volatility model: 
\begin{eqnarray*}
\frac{dS_t}{S_t} &=& (\mu+r)\,dt + f(Z_t)\,dB^1_t,  
\\
dZ_t &=&	\frac{1}{\eps} \alpha(Z_t)\,dt + \frac{1}{\sqrt{\eps}}\beta( Z_t)\, dB_t^2 , 
\label{eq:Z} 
\end{eqnarray*}
where $B^1$ and $B^2$  are Brownian motions, defined on a filtered probability space $(\Omega, \sF, \{\sF(t)\}_{t\geq 0},\P)$, with constant correlation coefficient $\rho\in(-1,1)$:  $d\< B^1, B^2 \>_t
=\rho\,dt$.
We assume that $f(z)$ is a smooth, bounded and strictly positive function, and that the stochastic volatility factor $Z_t$ is a fast mean-reverting process, meaning that the parameter $\eps>0$ is small, and that $Z$ is an ergodic process with a unique invariant distribution $\Phi$ that is independent of $\eps$. We refer to \cite[Chapter 3]{fpss-book} for further technical details and discussion. 
Additionally $r, \mu$ are positive constants, and $\alpha, \beta$ are smooth functions: examples will be specified later for computations. 

\subsection{Investment Problem}
The investor must choose a policy consisting of two adapted processes $L$ and $M$ that are nondecreasing and right-continuous with left limits, and $L_{0-} =M_{0-} = 0$. The control $L_t$ represents the cumulative dollar value of stock purchased up to time $t$, while $M_t$ is the cumulative dollar value of stock sold. Then, the wealth $X$ invested in the money market account and the wealth $Y$ invested in the stock follow
\begin{eqnarray*}
dX_t		&=&		 r X_t \,dt - dL_t + \( 1 - \lam \) dM_t ,\label{eq:X}\\
dY_t		&=&		(\mu +r)  Y_t\,dt + f(Z_t) Y_t \, dB_t^1 +  dL_t - dM_t. 
\end{eqnarray*}
The constant $\lambda \in (0,1)$ represents the proportional transaction costs for selling the stock.

Next, we define the \emph{solvency region}
\begin{equation}
\Sv \triangleq
\left\{(x,y);\:x+\,y > 0,\: x+(1-\lambda) \,y >
0\right\},\label{eq:solvency}
\end{equation}
which is the set of all positions, such that if the investor were forced to liquidate immediately, she would not be bankrupt. This leads to a definition that a policy $(L_s,M_s)\big|_{s\ge t}$ is \emph{admissible} for the initial position $Z_t=z$ and $(X_{t-}, Y_{t-})=(x,y)$ starting at time $t^{-}$, if $(X_s,Y_s)$ is in the closure of solvency region, $\overline{\Sv} $, for all $s\ge t$.
(Since the investor may choose to immediately rebalance his position, we have denoted the initial time $t-$). Let
$\A(t,x,y,z)$ the set of all such policies. Clearly, if $(x,y) \in\overline{\Sv}$ then we can always liquidate the position, and then hold the resulting cash position in the risk-free money market account.
It is easy to adapt the proof in \cite{ShreveSoner} (for the constant volatility case) to show that 
$\A(t,x,y,z) \neq \emptyset$ if and only if $(t,x,y,z) \in[0,\infty)\times\overline{\Sv}\times\R.$ 

We work with CRRA  or power utility functions $\U(w)$ defined on $\R_+$:
$$ \U(w) \define \frac{w^{1-\gamma}}{1-\gamma}, \qquad \gamma>0,\quad\gamma\ne1,$$
where $\gamma$ is the constant relative risk aversion parameter. We are interested in maximizing:
$$
\sup_{(L,M) \in \A(0,x,y)}\liminf\limits_{T\to\infty} \frac1T \log \U^{-1}\left(\E^{x,y,z}_0\left[ \U\left(X_T+Y_T-\lam Y_T^{+}\right)\right]\right), \quad (x,y,z) \in \overline{\Sv},\times\R,$$
where $\E^{x,y,z}_t[\cdot]\define\E[\cdot \big| X_{t-}=x,Y_{t-}=y, Z_t=z].$
This is a problem in optimizing long term growth. To see the economic interpretation note that 
the quantity $U^{-1}\left(\E^{x,y,z}_0\left[ U\left(X_T+Y_T-\lam Y_T^{+}\right)\right]\right)$ is the {\em certainty equivalent} of the terminal wealth $X_T+Y_T-\lam Y_T^{+}$. Hence if we can match this certainty equivalent with $(x+y-\lam y^{+})\e{(r+\delta^\eps)T}$ -- the investor's initial capital compounded at some rate $r+\delta^\eps$, then $\frac1T\log U^{-1}\left(\E^{x,y,z}_0\left[ U\left(X_T+Y_T-\lam Y_T^{+}\right)\right]\right)=r+\delta^\eps.$ 
For a survey and literature on this choice of objective function we refer to 
\cite{guasoni}.
This choice of optimization problem ensures the simplest HJB equation, which in this case turns out to be linear and time independent.

\subsection{HJB Equation}
Consider first the value function for utility maximization at a finite time horizon $T$:
$$ 
\V(t,x,y,z) = \sup_{(L,M) \in \A(t,x,y,z)}\E_t^{x,y,z} \[ \U\(X_T+Y_T-\lam Y_T^{+}\)\].
$$
From It\^o's formula it follows that
\begin{align*}
&d\V(t, X_t, Y_t, Z_t) \\
&=\left(\V_t + rX_t \V_x + \(\mu+r\) Y_t \V_y + \frac{1}{2}f^2(Z_t) Y_t^2 \V_{yy} + \frac{1}{\sqrt{\eps}} \rho f(Z_t) \beta(Z_t) Y_t \V_{yz}\right)dt\\
&+ \frac{1}{\eps}\left( \alpha(Z_t) \V_z + \frac{1}{2}\beta^2(Z_t) \V_{zz}^2 \right)dt + f(Z_t) Y_t \V_y\, dB_t^1 + \frac{1}{\sqrt{\eps}}\beta( Z_t )V_z \, dB_t^2\\
&+\left(\V_y - \V_x  \right) dL_t + \left( (1-\lam)\V_x - \V_y\right)dM_t.
\end{align*}
Since $\V$ must be a supemartingale, the $dt, dL_t$ and $dM_t$ terms must not be positive. It follows that $\V_y - \V_x\le0$ and $(1-\lam)\V_x - \V_y \le 0$. Alternatively, 
\begin{equation}
1\le\frac{\V_x}{\V_y}\le \frac{1}{1-\lam}.\label{NTdef}
\end{equation} 
We will define the no-trade ($\widehat{\mbox{NT}}$) region, associated with $\V$, to be the region where both of these inequalities are strict. Moreover, for the optimal strategy, $\V$ is a martingale, and so the $dt$ term above must be zero inside the $\widehat{\mbox{ NT}}$ region. Thus it will then satisfy the HJB equation
\begin{equation}
\max\left\{(\d_t+\D^\eps)\V , (\d_y - \d_x)\V, \(\(1-\lam\) \d_x - \d_y\)\V  \right\}=0, ~\V(T,x,y,z)=\U(x+y - \lam y^{+}),\label{eq:HJB-init}
\end{equation}
where 
\begin{align}
\D^\eps
		&=		r x \d_x + \(\mu+r\) y \d_y + \frac{1}{2}f^2(z) y^2 \d_{yy}^2  + \frac{1}{\sqrt{\eps}} \rho f(z) \beta(z) y \d_{yz}^2
					\\
					&+ \frac{1}{\eps}\( \alpha(z) \d_z + \frac{1}{2}\beta^2(z) \d_{zz}^2 \).
\end{align}

The fact that $\V$ is a viscosity solution of \eqref{eq:HJB-init} is standard, and a similar proof can be found for example in \cite{ShreveSoner}, and thus will be omitted here. We will furthermore assume that the viscosity solution $\V$ of \eqref{eq:HJB-init} is in fact a classical solution, that is we will assume that it is sufficiently smooth. It can be shown that $\V$ is smooth inside each of three regions: the $\widehat{\mbox{NT}}$, and the regions where $ (\d_y - \d_x)\V=0,$ and $\(\(1-\lam\) \d_x - \d_y\)\V=0$. The assumption that it is also smooth on the {\em boundary} of the $\widehat{\mbox{NT}}$ is the smooth fit assumption, which is very common; see, for instance, \cite{GoodmanOstrov}.

Next, we look for a solution of the HJB equation \eqref{eq:HJB-init} of the form
\begin{equation}
V(t,x,y,z) = x^{1-\gamma}{\v}	\(\zeta,z\)\e{ (1-\gamma)(r+\delta^\eps) (T-t) } ,\qquad \zeta=\frac{y}x,
\label{eq:substitution}
\end{equation}
where $\delta^\eps$ is a constant, and the function ${\v}$ is to be found. However, we will not impose the final time condition on $V$.
For now, we will only assume that it is smooth and $\abs{{\v}}$ is bounded away from zero. We will define the NT region (associated with $V$) as the region where $\(\d_t+\D^\eps\)V=0.$ Additionally, we will assume that for any point $(t,x,y,z)$ in the NT region, the ratio $y/x$ is bounded. We note that $V$ is not equivalent to the value function $\V$, since we have not imposed the final time condition $V(T,x, y, z) = \U(x+y-\lam y^{+}).$ In fact there is no reason to believe that the final time condition can be satisfied if $V$ is given by \eqref{eq:substitution}. %

However, if we find a constant $C$ such that $\abs{\V}\le C\abs{V}$, then it would follow that $\delta^\eps$ is the optimal growth rate and the NT region for the long-term optimal growth problem can be defined as the region where $\(\d_t+\D^\eps\)V=0.$ In other words, 
\begin{eqnarray*}
\liminf\limits_{T\to\infty} \frac1T \log \U^{-1}\left(\V(0,x,y,z)\right)  &=& \liminf\limits_{T\to\infty} \frac1T \log \U^{-1}\left(V(0,x,y,z)\right)\\
& = &\liminf\limits_{T\to\infty} \frac1T\frac{ \log V(0,x,y,z)}{1-\gamma}\\
&=&r+\delta^\eps. 
\end{eqnarray*}

We will now show that there exists a constant $C$ such that $\abs{\V}\le C\abs{V}$. Indeed, note that the utility function $\U$ is homogeneous of degree $ 1-\gamma$, that is $\U(w)=w^{1-\gamma}\U(1)$, it follows that
$$\U\(X_T+Y_T-\lam Y_T^{+}\) = X_T^{1-\gamma}\U\(1 + \frac{Y_T}{X_T} -\lam \(\frac{Y_T}{X_T}\)^{+}\).$$ 
By our assumption, $Y_T/X_T$ is bounded, being inside the NT region. Hence, there exists a constant $C$ such that 
\begin{align}
\abs{\V(T,x,y,z)} =\abs{\U\(x +y-\lam y^{+}\)} \le C \abs{V(T,x,y,z)} =x^{1-\gamma}C{\v}\(\frac{y}{x},z\).
\label{eq:final-cond}
\end{align}
Since both $V$ and $\V$ solve the HJB equation \eqref{eq:HJB-init}, it follows by a comparison theorem that $\abs{\V}\le C\abs{V}$ everywhere. For the reader convenience, we have sketched the proof of it in  Appendix \ref{sec:comparisonThm}.

Inserting the transformation \eqref{eq:substitution} into \eqref{eq:HJB-init} leads to the following equation for $({\v},\delta^\eps)$:
\begin{align}
\max\left\{\frac1\eps\L_0 + \frac1{\sqrt{\eps}}\L_1 + \(\L_2-\(1-\gamma\)\delta^\eps\cdot\)\,,~\B,~ \S   \right\}\v =0,
\label{eq:HJB-reduced}
\end{align}
where we define the operators in the NT region by
\begin{equation}
\L_0 =  \frac{1}{2}\beta^2(z) \d_{zz}^2+\alpha(z) \d_z, \qquad
\L_1 = \rho f(z) \beta(z) \zeta \d_{\zeta z}^2, \qquad
\L_2	=  \frac12 f^2(z)\zeta^2\d_{\zeta\zeta} + \mu\zeta \d_{\zeta},   \label{eq:Ls}
\end{equation}
and the buy and sell operators by
\begin{align}
\B   &=  \(1+\zeta\) \d_{\zeta} - (1-\gamma)\cdot\, ,\label{eq:B}\\
\S  & = \(\frac1{1-\lam} +\zeta\) \d_{\zeta} - (1-\gamma)\cdot\,,\label{eq:S}
\end{align}
respectively.
For future reference, we also define their derivatives
\begin{align}
\B'   &=  \d_\zeta \B = \(1+\zeta\) \d_{\zeta\zeta}+\gamma \d_{\zeta}\, ,\label{eq:B'}\\
\S'  & =  \d_\zeta \S=\(\frac1{1-\lam} +\zeta\) \d_{\zeta\zeta} + \gamma \d_{\zeta}\,.\label{eq:S'}
\end{align}

\subsection{Free Boundary Formulation \& Eigenvalue Problem\label{freebf}}
We will look for a solution to the variational inequality \eqref{eq:HJB-reduced} in the following free-boundary form. The NT region is defined by \eqref{NTdef}, but for the function $V$. Using the transformation \eqref{eq:substitution}, this translates to 
$$1+\zeta <(1-\gamma)\(\frac{\v}{\v_{\zeta}}\) < \frac{1}{1-\lam} +\zeta$$
for $\v(\zeta,z)$. Similar to the case with constant volatility, we assume that there exists a no-trade region, within which $\(\frac1\eps\L_0 + \frac1{\sqrt{\eps}}\L_1 + \(\L_2-\(1-\gamma\)\delta^\eps\cdot\)\)\v=0,$ with boundaries $\l^\eps(z)$ and $u^\eps(z)$.
We write this region as 
$$\min\{\l^\eps(z),u^\eps(z)\}<\zeta< \max\{\l^\eps(z),u^\eps(z)\},$$
where $\l^\eps(z)$ and $u^\eps(z)$ are free boundaries to be found. In typical parameter regimes, we will  have $0<\l^\eps(z)<u^\eps(z)$, so we can think of them as lower and upper boundaries respectively, with $\l^\eps$ being the buy boundary, and $u^\eps$ the sell boundary. (The other two possibilities are that $\l^\eps<u^\eps<0$ with $\l^\eps$ being the buy boundary, and $u^\eps$ the sell boundary, or that $\l^\eps<u^\eps<0$ with $\l^\eps$ being the sell boundary, and $u^\eps$ the buy boundary. Under a constant volatility model these cases can be categorized explicitly in term of the model parameters: see Remark \ref{rem:cases}).

Inside this region we have from the HJB equation \eqref{eq:HJB-reduced} that 
\begin{equation}
\left(\frac1\eps\L_0 + \frac1{\sqrt{\eps}}\L_1 + \(\L_2-\(1-\gamma\)\delta^\eps\)\right) {\v} =0,\qquad
\zeta\in(\l^\eps(z),u^\eps(z)).\label{NTeqn}
\end{equation}
The free boundaries $\l^\eps$ and $u^\eps$ are determined by continuity of the first and second derivatives of $\v$ with respect to $\zeta$, that is looking for a $C^2$ solution. In the buy region,
\begin{equation}
\B\v=0 \qquad\mbox{in } \zeta<\l^\eps(z),\label{buyeqn} 
\end{equation}
and so the smooth pasting conditions at the lower boundary are
\begin{align}
\B\v\mid_{\l^\eps(z)} &= \(1+\l^\eps(z)\)\v_\zeta(\l^\eps(z)) - (1-\gamma)\v(\l^\eps(z)) =  0,\label{lbcn1}\\
\B'\v\mid_{\l^\eps(z)} &\define\(1+\l^\eps(z)\)\v_{\zeta\zeta}(\l^\eps(z)) + \gamma\v_\zeta(\l^\eps(z))  =  0.\label{lbcn2}
\end{align}
In the sell region, the transaction cost enters and we have:
\begin{equation}
\S\v=0 \qquad \mbox{in } \zeta>u^\eps(z). \label{selleqn}
\end{equation}
Therefore the sell boundary conditions are:
\begin{align}
\S\v\mid_{u^\eps(z)} &=\(\frac1{1-\lam} +u^\eps(z)\)\v_\zeta(u^\eps(z)) - (1-\gamma)\v(u^\eps(z)) =  0,\label{ubcn1}\\
\S'\v\mid_{u^\eps(z)} &\define\(\frac1{1-\lam}+u^\eps(z)\)\v_{\zeta\zeta}(u^\eps(z)) + \gamma\v_\zeta(u^\eps(z))  =  0.\label{ubcn2}
\end{align}

We note that \eqref{NTeqn}, \eqref{buyeqn} and \eqref{selleqn} are homogeneous equations with homogeneous boundary conditions \eqref{lbcn1}, \eqref{lbcn2}, \eqref{ubcn1} and \eqref{ubcn2}, and so zero is a solution. However the constant $\delta^\eps$ is also to be determined, and in fact it is an eigenvalue found to exclude the trivial solution and give the optimal long-term growth rate. In the next section, we construct an asymptotic expansion in $\eps$ for this eigenvalue problem using these equations.

\section{Fast-scale Asymptotic Analysis}\label{sec:asymptotic}
We look for an expansion for the value function
\begin{equation}
\v =  v^{\lam,0} + \sqrt{\eps}\, v^{\lam,1} + \eps  v^{\lam,2} + \cdots,   \label{eq:v-expansion}
\end{equation}
as well as for the free boundaries
\begin{equation}
\l^\eps = \l_0 + \sqrt{\eps}\,\l_1 + \eps\l_2 + \cdots,\qquad u^\eps = u_0 + \sqrt{\eps}\,u_1 + \eps u_2 + \cdots, \label{boundaryexpn}
\end{equation}
and the optimal long-term growth rate
\begin{equation}
\delta^\eps = \delta_0 + \sqrt{\eps}\delta_1 + \cdots, 
\label{eq:delta-expansion}
\end{equation}
which are asymptotic as $\eps\downarrow0$.

Crucial to this analysis is the Fredholm alternative (or centering condition) as detailed in \cite{fpss-book}. In preparation, we will use the notation $\< \cdot\>$ to denote the expectation with respect to the invariant distribution $\Phi$ of the process $Z$, namely
\begin{equation}
\<g\> \define\int g(z)\Phi(dz).
\label{eq:bracket}
\end{equation}
The Fredholm alternative tells us that a Poisson equation of the form 
$$\L_0 v+ \chi   = 0$$
has a solution $v$ only if the solvability condition $\< \chi \>   = 0$ is satisfied, and we refer for instance to \cite[Section 3.2]{fpss-book} for technical details.

It is also convenient to introduce the differential operators
\begin{equation}
D_k = \zeta^k\frac{\partial^k}{\partial\zeta^k}, \qquad k=1,2, \cdots, \label{Dkdef}
\end{equation}
in terms of which the operators $\L_1$ and $\L_2$ in \eqref{eq:Ls} are
$$ \L_1 = \rho f(z) \beta(z) \d_{z}D_1, \qquad \L_2	=  \frac12 f(z)^2D_2 + \mu D_1. $$
In the following, a key role will be played by the squared-averaged volatility $\sigb$ defined by
\begin{equation}
\sigb^2=\<f^2\>. \label{sigbdef}
\end{equation}
The principal terms in the expansions will be related to the {\em constant volatility} transaction costs problem, and we define the operator $\lnt(\sigma;\delta)$ that acts in the no trade region by
\begin{equation}
\lnt(\sigma;\delta) = \frac12 \sigma^2D_2+\mu D_1 -\(1-\gamma\)\delta\cdot\,,\label{lntdef}
\end{equation}
and it is written as a function of the parameters $\sigma$ and $\delta$. 

The zero-order terms in each of the asymptotic expansions \eqref{eq:v-expansion}, \eqref{boundaryexpn} and \eqref{eq:delta-expansion} are known and will be re-derived in Section \ref{sec:buildSol}.
In the rest of this section, we calculate the next terms in the above asymptotic expansion in the case of fast-scale stochastic volatility. 

\subsection{Power expansion inside the NT region}\label{sec:power-expan}
In this subsection we will concentrate on constructing the expansion inside the NT region $\zeta\in\(l^\eps(z),u^\eps(z)\)$, where \eqref{eq:HJB-reduced} holds.
We now insert the expansion \eqref{eq:v-expansion} and match powers of $\eps$. 

The terms of order $\eps^{-1}$ lead to $\L_0 v^{\lam,0}=0$.
Since the $ \L_0 $ operator takes derivatives in $z$, we seek a solution of the form $v^{\lam,0} = v^{\lam,0}(\zeta)$, independent of $z$.

At order $\eps^{-1/2}$, we have  
$\L_1 v^{\lam,0} + \L_0 v^{\lam,1}  = 0$.
But since $\L_1$ takes a derivative in $z$, $\L_1 v^{\lam,0}  = 0$, and so $\L_0 v^{\lam,1}  = 0$.
Again, we seek a solution of the form $v^{\lam,1}=v^{\lam,1}(\zeta)$
that is independent of $z$.

The terms of order one give
$$\(\L_2-\(1-\gamma\)\delta_0\)v^{\lam,0} + \L_1 v^{\lam,1} + \L_0 v^{\lam,2}  = 0.$$
Since we have that $\L_1$ takes derivatives in $z$, and $  v^{\lam,1}  $ is independent of $z$, we have that 
\begin{equation}
\(\L_2-\(1-\gamma\)\delta_0\)v^{\lam,0}  + \L_0 v^{\lam,2}   = 0.\label{v2eqn}
\end{equation}
This is a Poisson equation for $v^{\lam,2}$ with $\< \(\L_2-\(1-\gamma\)\delta_0\) \> v^{\lam,0}   = 0$ as the
solvability condition. We observe that
$$\<\(\L_2-\(1-\gamma\)\delta_0\cdot\)\> =  \lnt(\sigb;\delta_0),$$ 
where $\sigb$ is the square-averaged volatility defined in \eqref{sigbdef}, and $\lnt$ is the constant volatility no trade operator defined in \eqref{lntdef}. Then we have
\begin{equation}
 \lnt(\sigb;\delta_0) v^{\lam,0}   = 0,
\label{eq:L_Mer}
\end{equation}
which, along with boundary conditions we will find in the next subsection, will determine $v^{\lam,0}$.

To find the equation for the next term $v^{\lam,1}$ in the approximation, we proceed as follows. We write the first term of \eqref{v2eqn} as
\begin{align*}
\(\L_2-\(1-\gamma\)\delta_0\)v^{\lam,0}   &= \(\(\L_2-\(1-\gamma\)\delta_0\) - \lnt(\sigb;\delta_0)\) v^{\lam,0}= \frac12\(f^2(z) - \sigb^2\)
D_2v^{\lam,0}.
\end{align*}
The solutions of \eqref{v2eqn} are given by
\begin{align}
v^{\lam,2}=-\L_0^{-1}\L_2v^{\lam,0}=-\frac12\L_0^{-1}\(f^2(z) - \sigb^2\)D_2v^{\lam,0}=-\frac12\ \(\phi(z)+c(\zeta) \)D_2v^{\lam,0},~~~~~
\label{eq:v2}
\end{align}
where $c(\zeta)$ is independent of $z$, and %
$\phi(z)$ is a solution to Poisson equation
\begin{align}
\L_0 \phi(z)   &= f^2(z) - \sigb^2,
\label{eq:Poisson}
\end{align}

Continuing to the order $\sqrt{\eps}$ terms, we obtain
\begin{align*}
\(\L_2-\(1-\gamma\)\delta_0\)v^{\lam,1} + \L_1 v^{\lam,2} + \L_0 v^{\lam,3}-(1-\gamma)\delta_1v^{\lam,0}   &= 0.
\end{align*} 
Once again, this is a Poisson equation for $v^{\lam,3}$ whose centering condition implies that
$$\< \(\L_2-\(1-\gamma\)\delta_0\) \> v^{\lam,1} + \< \L_1 v^{\lam,2} \>-(1-\gamma)\delta_1v^{\lam,0}    = 0.$$
From \eqref{eq:v2}, it follows that
\begin{equation}
\lnt(\sigb;\delta_0) v^{\lam,1} - (1-\gamma)\delta_1v^{\lam,0}   = - \< \L_1 v^{\lam,2} \> = \frac12\<\L_1 \phi\> D_2v^{\lam,0}= \frac12\rho\<\beta f\phi'\>  D_1D_2v^{\lam,0}.
\label{eq:HJB1}  
\end{equation} 
We define
\begin{equation}
V_3 = -\frac12\rho\<\beta f\phi'\>. 
\label{eq:V3}
\end{equation}
Then we write the equation \eqref{eq:HJB1} as 
\begin{equation}
\lnt (\sigb;\delta_0) v^{\lam,1}   = -V_3D_1D_2v^{\lam,0}+(1-\gamma)\delta_1v^{\lam,0}. \label{v1eqn}
\end{equation}

\subsection{Boundary Conditions}\label{sec:bndry}
So far we have concentrated on the PDE \eqref{eq:HJB-reduced} in the NT region.
We now insert the expansions \eqref{eq:v-expansion} and \eqref{boundaryexpn} into the boundary conditions \eqref{lbcn1}--\eqref{ubcn2}.
The terms of order one from \eqref{lbcn1}  and \eqref{lbcn2} give
\begin{equation}
\B v^{\lam,0}\mid_{\l_0}=0, \qquad \mbox{and } \qquad \B' v^{\lam,0}\mid_{\l_0}=0, \label{v0bcnl}
\end{equation}
while the terms of order one from \eqref{lbcn1}  and \eqref{lbcn2} give
\begin{equation}
\S v^{\lam,0}\mid_{u_0}=0, \qquad \mbox{and } \qquad \S' v^{\lam,0}\mid_{u_0}=0, \label{v0bcnu}
\end{equation}
Since $v^{\lam,0}$ is independent of $z$, these equations imply that $\l_0$ and $u_0$ are also independent of $z$ (they are constants).

Taking the order $\sqrt{\eps}$ terms in \eqref{lbcn1} gives
$$ (1+\l_0)\(v_{\zeta} ^{\lam,1}(\l_0)+\l_1v_{\zeta\zeta} ^{\lam,0}(\l_0)\) + \l_1v_{\zeta} ^{\lam,0}(\l_0) -(1-\gamma)\(v^{\lam,1}(\l_0) + \l_1v_{\zeta}^{\lam,0}(\l_0)\)=0.$$
Using the fact that $\B v^{\lam,0}\mid_{\l_0}=0$, we see the terms in $\l_1$ cancel, and we obtain 
\begin{equation}
 \B v^{\lam,1}\mid_{\l_0}=0,\label{eq:boundary3}
 \end{equation}
which is a mixed-type boundary condition for $v^{\lam,1}$ at the boundary $\l_0$.

From the order $\sqrt{\eps}$ terms in \eqref{lbcn2}, we obtain
\begin{equation}
\l_{1}\( v_{\zeta\zeta}^{\lam,0}(\l_0)  +\(1+\l_{0}\)v_{\zeta\zeta\zeta}^{\lam,0}(\l_0) + \gamma v_{\zeta\zeta}^{\lam,0}(\l_0) \)+\[\(1+ \l_{0}\)v_{\zeta\zeta} ^{\lam,1}(\l_0) +\gamma v_{\zeta} ^{\lam,1}(\l_0)\]=0,
\label{eq:boundary3.1}
 \end{equation}
and so, as $v^{\lam,1}$ does not depend on $z$,  $\l_1$ is also a constant (independent of $z$) 
given by
\begin{equation}
\l_1=  -\(\frac{%
\B' v^{\lam,1}\mid_{\l_0}}
{\(1+\l_{0}\)v_{\zeta\zeta\zeta}^{\lam,0}(\l_0) + (1+\gamma)v_{\zeta\zeta}^{\lam,0}(\l_0)}\).
\label{eq:l1}
\end{equation}

Similar calculations can be performed on the (right) sell boundary $\u^\eps\approx \u_{0} + \sqrt{\eps}\u_{1},$ where $\S {\v}=0$. The analogous equations to \eqref{eq:boundary3} and \eqref{eq:l1} are 
\begin{align}
&%
\S v^{\lam,1}\mid_{u_0}=0.
\label{eq:boundary5}
\\
&\u_1=  -\(\frac{
\S' v^{\lam,1}\mid_{u_0}
}{\(\frac{1}{1-\lam}+\u_{0}\)v_{\zeta\zeta\zeta}^{\lam,0}(u_0) + (1+\gamma) v_{\zeta\zeta}^{\lam,0}(u_0)}\).\label{eq:u1}
\end{align}
Note that \eqref{eq:boundary5} is a mixed-type boundary condition for $v^{\lam,1}$ at the boundary $u_0$, and \eqref{eq:u1} determines the constant correction term $u_1$ to the sell boundary.

\subsection{Determination of  $\delta_1$}\label{sec:v1}
The next term $v^{\lam,1}$ in the asymptotic expansion solves the ODE
\eqref{v1eqn}, with boundary conditions \eqref{eq:boundary3} and \eqref{eq:boundary5}, but we also need to find $\delta_1$ which appears in the equation. In fact, the Fredholm solvability condition for this equation determines $\delta_1$, and so we look for the solution $\w$ of the homogeneous adjoint problem. 

To do that we first multiply both sides of \eqref{v1eqn} by $\w$ and integrate from $\l_0$ to $\u_0$:
\begin{equation}
\int_{\l_0}^{\u_0}  \w \lnt v^{\lam,1}  d\zeta   = -V_3\int_{\l_0}^{\u_0}D_1D_2v^{\lam,0}\w d\zeta+(1-\gamma)\delta_1\int_{\l_0}^{\u_0}v^{\lam,0}\w d\zeta. \label{wv}
\end{equation}
Integration by parts gives
\begin{equation}
\int_{\l_0}^{\u_0}  \w\lnt v^{\lam,1}d\zeta = \int_{\l_0}^{\u_0}  v^{\lam,1}\lnt^{*}\w\,d\zeta
+\left[\frac{\sigb^2}{2}v_{\zeta}^{\lam,1}\zeta^2\w - \frac{\sigb^2}{2} (\zeta^2\w)' v^{\lam,1} + \mu\zeta\w v^{\lam,1}\right]_{\l_0}^{\u_0}, \label{eq:by-parts}
\end{equation}
where $\lnt^{*}=\lnt^{*}(\sigb;\delta_0)$ is the adjoint operator to $\lnt$:
\begin{align}
\lnt^{*}(\sigb;\delta_0)(w) =\frac12\sigb^2 \d_{\zeta\zeta} \left(\zeta^2 w\right)- \mu \d_{\zeta}(\zeta w) - (1-\gamma)\delta_0 w.
\label{eq:Merton-adj}
\end{align}
We set $\w$ to satisfy 
\begin{equation}
\lnt^{*}(\sigb;\delta_0)(w)=0, \label{weqn}
\end{equation} 
and, to cancel the boundary terms in \eqref{eq:by-parts}, 
the boundary conditions
\begin{equation}
\l_0 \w'(\l_0)- k_-\w(\l_0)=0, \qquad u_0 \w'(\u_0)-k_+\w(\u_0)=0, 
\label{eq:bnd1}
\end{equation}
where we define the constants
\begin{equation}
k_{\pm} \define (1-\gamma)\pi_{\pm} +  (k-2), \quad \mbox{and} \quad k\define\frac{\mu}{\frac12\sigb^2}. \label{kdefs}
\end{equation}
\begin{lemma}\label{wlemma}
The solution $w$ to the adjoint equation \eqref{weqn} with boundary conditions \eqref{eq:bnd1} is, up to a multiplicative constant, given by
\begin{equation}
w(\zeta) = \zeta^{k-2}v^{\lam,0}(\zeta), \label{wsol}
\end{equation}
where $k$ was defined in \eqref{kdefs}.
\end{lemma}
\begin{proof}
Making the substitution \eqref{wsol} into \eqref{weqn} leads to the equation \eqref{eq:L_Mer} satisfied by $v^{\lam,0}$. Similarly inserting \eqref{wsol} into the boundary conditions \eqref{eq:bnd1} leads to the boundary conditions \eqref{v0bcnl} and \eqref{v0bcnu} satisfied by $v^{\lam,0}$. The conclusion follows.
\end{proof}

Now the left hand side of \eqref{wv} is zero, and so we find that $\delta_1$ is given by \begin{align}
\delta_1 = \frac{V_3}{(1-\gamma)}\frac{\int_{\l_0}^{\u_0}\w D_1D_2v^{\lam,0} d\zeta}{\int_{\l_0}^{\u_0}\w v^{\lam,0}d\zeta}.
\label{eq:delta1}
\end{align}
Note that $\delta_1$ is well defined, as the undetermined multiplicative constant of $v^{\lam,0}$ cancels in the ratio.

\subsection{Summary of the Asymptotics}\label{sec:main-fast}
To summarize, we have sought the zeroth and first order terms in the expansions \eqref{eq:v-expansion}, \eqref{boundaryexpn} and \eqref{eq:delta-expansion} for $(v^{\lam,\eps},\l^\eps,u^\eps,\delta^\eps)$.
The principal terms are found from the eigenvalue problem described by ODE \eqref{eq:L_Mer}, with boundary and free boundary conditions \eqref{v0bcnl}-\eqref{v0bcnu}:
\begin{align*}
&  \lnt(\sigb;\delta_0) v^{\lam,0}   = 0, \qquad  \l_0\leq\zeta\leq u_0,\\
& \B v^{\lam,0}\mid_{\l_0}=0, \qquad \mbox{and } \qquad \B' v^{\lam,0}\mid_{\l_0}=0\,;
& \S v^{\lam,0}\mid_{u_0}=0, \qquad \mbox{and } \qquad \S' v^{\lam,0}\mid_{u_0}=0.
\end{align*}
The next term in the asymptotic expansion of the boundaries of the NT region, and of the optimal long-term growth rate $\l_1,\u_1,$ and $\delta_1$ respectively, are given by \eqref{eq:l1}, \eqref{eq:u1} and \eqref{eq:delta1}, and $v^{\lam,1}$
solves the ODE \eqref{v1eqn}, with boundary conditions \eqref{eq:boundary3} and \eqref{eq:boundary5}:
\begin{align*}
& \lnt (\sigb;\delta_0) v^{\lam,1}   = -V_3D_1D_2v^{\lam,0}+(1-\gamma)\delta_1v^{\lam,0}, \qquad  \l_0<\zeta< u_0,\\
& \B v^{\lam,1}\mid_{\l_0}=0, \qquad \mbox{and } \qquad \S v^{\lam,1}\mid_{u_0}=0.
\end{align*}
We describe the essentially-explicit solutions to these problems in the next section.

\section{Building the Solution}\label{sec:buildSol}
In the previous section we have established that $(v^{\lam,0},\delta_0)$ solve the {\em constant volatility} optimal growth rate with transaction costs problem, which is described in \cite{DumasLuciano}, but using the averaged volatility $\sigb$, where $\sigb^2=\langle f^2\rangle$. 
In this section, we review how to find $(v^{\lam,0},\delta_0)$, and then use them to build the stochastic volatility corrections $(v^{\lam,1},\delta_1)$. %

\subsection{Building $v^{\lam,0}$ and $\delta_0$} \label{sec:buildSol1}
We denote by $(V_0(\zeta;\sigma),\Delta_0(\sigma),L_0(\sigma),U_0(\sigma))$ the solution to the constant volatility problem with volatility parameter $\sigma$ and corresponding eigenvalue $\Delta_0$, and so 
\begin{equation}
v^{\lam,0}(\zeta)=V_0(\zeta;\sigb), \quad \delta_0=\Delta_0(\sigb), \quad\mbox{and }\quad (\l_0,u_0)=(L_0(\sigb),U_0(\sigb)). \label{v0sol}
\end{equation}

\begin{assumption}\label{A1}
Without loss of generality assume that $\mu>0$. The case $\mu<0$ can be handled similarly to the current case. The case $\mu=0$ is not interesting, as in this case one would not hold the risky stock at all. 
We also assume that the optimal proportion of wealth invested into the risky stock in case of constant volatility $\sig$ and zero transaction costs is less than 1: 
\begin{equation}
\pi_M:=\frac{\mu}{\gamma\sig^2}<1. \label{mertonbound}
\end{equation}
We will refer to $\pi_M$ as the Merton proportion. %
\end{assumption}
\begin{remark}\label{rem:cases}
It turns out that under the assumption \eqref{mertonbound}, we will have $0<L_0<U_0$. The other two cases when $\mu<0$ and when $\frac{\mu}{\gamma\sig^2}>1$ can be handled similarly.  If $\mu<0$, we would have $L_0<U_0<0$, with $L_0$ being the buy boundary, and $U_0$ the sell boundary. 
If we had $\pi_M>1$, then $L_0<U_0<0$ and $L_0$ is the sell boundary, and $U_0$ the buy boundary. For $\eps>0$ small enough, the same is true for $\l^\eps(z)$ and $u^\eps(z)$. 
The final case when $\pi_M=1$ is not interesting either, since in this case, all wealth will be invested into stock, and no trading will be necessary, except possibly at the initial time.
Also, note that under these assumptions, we are assured that the NT region is non-degenerate.
\end{remark}

In preparation, we define the following quantities. Given $\Delta_0$, we define %
\begin{equation}
\theta_{\pm}=\theta_{\pm}(\Delta_0)%
\mbox{  as the roots of } \frac12\sigma^2\theta^2 + \left( \mu - \frac12\sigma^2\right) \theta  - (1-\gamma)\Delta_0=0,
\label{eq:theta}
\end{equation}
and let 
\begin{equation}
\pi_{\pm}=\pi_{\pm}(\Delta_0) %
\mbox{ be the roots of } \frac12\gamma\sigma^2\pi^2 -\mu\pi +\Delta_0 =0,
\label{eq:pi}
\end{equation}
where in both cases, we will suppress the dependency on $\Delta_0$. Additionally,
let $$L_0 \define \frac{\pi_{-}}{1-\pi_{-}}, \qquad U_0 \define \(\frac1{1-\lam}\)\frac{\pi_{+}}{1-\pi_{+}},$$ where we have again suppressed the dependency on $\Delta_0$, and we also define
\begin{equation} 
k_\l \define \frac{1+{L_0}}{1-\gamma} \qquad \mbox{and } \qquad k_u \define \frac{\frac1{1-\lam}+{U_0}}{1-\gamma}. \label{klku}
\end{equation}

\begin{proposition}\label{prop:fast}
The function $V_0(\zeta)$ is given by
$$ V_0(\zeta)=c_{+}v_+(\zeta) + c_{-}v_-(\zeta), \qquad \mbox{with } \quad c_{\pm} \define  v_\mp({L_0}) - k_\l v_\mp'({L_0}), $$
where, given $(\mu,\sigma, \lambda, \gamma)$, there are two cases:
\begin{description}
\item[Real Case:] The eigenvalue $\Delta_0$ is a real root of the algebraic equation
\begin{equation}
 \left( \frac{\theta_{+}}{\pi_{-}} + \frac{\theta_{-}}{\pi_{+}} -(1-2\gamma) \right){L_0}^{(\theta_+-\theta_-)} - 
\left( \frac{\theta_{+}}{\pi_{+}} + \frac{\theta_{-}}{\pi_{-}} -(1-2\gamma) \right)U_0^{(\theta_+-\theta_-)} =0,\label{transcendreal}
\end{equation}
{\em and} $\theta_\pm(\Delta_0)$ in \eqref{eq:theta} are real and distinct. Then $v_\pm(\zeta) = \zeta^{\theta_\pm}.$

\item[Complex Case:] Otherwise, $\Delta_0$ is the real root of the transcendental equation
\begin{align}
\!\!\!\!\!\!\!\!\!\!\!\!\theta_i\left(\frac{k_\l}{{L_0}} - \frac{k_u}{U_0}  \right)-\left[\left(\frac{k_\l}{{L_0}}\theta_r - 1 \right)\left(\frac{k_u}{U_0}\theta_r - 1\right)+ \frac{\theta_i^2k_uk_\l}{U_0{L_0}}\right]\tan\(\theta_i\log\(\frac{U_0}{L_0}\)\)=0,
\label{eq:complex-eta}
\end{align}
where $\theta_{r}(\Delta_0), \theta_i(\Delta_0)$ the real and the imaginary parts of $\theta_{+}(\Delta_0)$.
Then 
\begin{equation}
 v_+(\zeta)=\zeta^{\theta_r}\cos(\theta_i\log\zeta), \qquad v_-(\zeta)=\zeta^{\theta_r}\sin(\theta_i\log\zeta). \label{vpmcomplex}
 \end{equation}
\end{description}

\end{proposition}

\begin{proof}
We have that $V_0$ solves the ODE $\lnt(\sigma;\Delta_0)V_0=0$ in the NT region: %
\begin{equation}
\frac12\sigma^2\zeta^2V_0'' + \mu \zeta V_0'  -(1-\gamma)\Delta_0 V_0=0, \qquad 0<{L_0}\leq\zeta\leq U_0, \label{v0eqn}
\end{equation}
with boundary conditions 
\begin{align}
\(1+ {L_0} \)V_0'({L_0})  - (1-\gamma)V_0({L_0}) &=0,\label{eq:boundary1}\\
\(1+ {L_0} \)V_0''({L_0})  + \gamma V_0'({L_0}) &=0 \label{eq:boundary2}
\end{align}
at the lower boundary ${L_0}$ and analogous conditions  
\begin{align}
\(\frac{1}{1-\lam} + U_0 \)V_0'(U_0)  - (1-\gamma)V_0(U_0) &=0,\label{eq:boundaryu1}\\
\(\frac{1}{1-\lam}+ U_0 \)V_0''(U_0)  + \gamma V_0'(U_0) &=0\label{eq:boundaryu2}
\end{align}
at the upper boundary $U_0$. 
This is a free-boundary problem with two undetermined boundaries and two conditions on each
boundary. We note that $V_0\equiv0$ is a solution of \eqref{v0eqn} and the boundary conditions \eqref{eq:boundary1}--\eqref{eq:boundaryu2}, but that the long-term growth rate $\Delta_0$ also has to be found. In fact it will be determined as an eigenvalue that eliminates the trivial solution.

First, substituting from \eqref{eq:boundary1} and \eqref{eq:boundary2} into \eqref{v0eqn} at ${L_0}$, we have  
\begin{align}
-\frac12\sigma^2(1-\gamma)\gamma\frac{{L_0}^2}{(1+{L_0})^2} V_0 + \mu(1-\gamma)\frac{{L_0}}{1+{L_0}}V_0 - (1-\gamma)\Delta_0 V_0=0.
\label{eq:pi2}
\end{align}
Then for a non-trivial $V_0$, equation \eqref{eq:pi2} becomes the quadratic equation
\begin{equation}
\frac12\gamma\sigma^2\pi_{-}^2 -\mu\pi_{-} +\Delta_0 =0, \qquad \mbox{where } \quad \pi_{-}\define \frac{{L_0}}{1+{L_0}}.
\label{eq:pim}
\end{equation}
By substituting \eqref{eq:boundaryu1} and \eqref{eq:boundaryu2} into \eqref{v0eqn} at ${U_0}$, we obtain the same equation 
\begin{equation}
\frac12\gamma\sigma^2\pi_{+}^2 -\mu\pi_{+} +\Delta_0 =0, \qquad \mbox{where } \quad \pi_{+} \define \frac{{U_0}}{\frac1{1-\lam}+{U_0}}.
\label{eq:pip}
\end{equation}
That is, $\pi_{\pm}$ are the two roots of the same quadratic \eqref{eq:pi}.

Next, let $v_+(\zeta)$ and $v_-(\zeta)$ be the two independent solutions of the second-order ODE \eqref{v0eqn}, so that the general solution is 
\begin{equation}
V_0=c_{+}v_+ + c_{-}v_-, \label{v0formula}
\end{equation}
for some constants $c_{\pm}$. Inserting this form into the boundary conditions \eqref{eq:boundary1} and \eqref{eq:boundaryu1}, and using the definitions in \eqref{klku} gives the linear system
\begin{equation}
M\left(\begin{array}{c}
c_{+}\\c_{-}\end{array}\right) = 0, \qquad \mbox{where }\quad M=
\left(\begin{array}{ll}
v_+({L_0}) - k_\l v_+'({L_0}) & v_-({L_0}) - k_\l v_-'({L_0})\\
v_+(u) - k_uv_+'(u) & v_-(u) - k_uv_-'(u)
\end{array}\right).\label{mtx1}
\end{equation}
Then, for a non-trivial solution, we require the determinant of $M$ to be zero, which leads to
\begin{align}
&\(v_+({L_0}) - k_\l v_+'({L_0})\)\(v_-(U_0) - k_uv_-'(U_0)\) \\
&- \(v_-({L_0}) - k_\l v_-'({L_0})\)\(v_+(U_0) - k_uv_+'(U_0)\) = 0. \label{determinanteqn}
\end{align}
Note that \eqref{determinanteqn} is an algebraic equation for the optimal long-term growth rate constant $\Delta_0$, where each term in the expression depends on $\Delta_0$ through \eqref{v0eqn}, \eqref{eq:pim} and \eqref{eq:pip}.

As $V_0$ will only be determined up to a multiplicative constant, we can choose 
\begin{equation}
c_{+} \define v_-({L_0}) - k_\l v_-'({L_0}),\qquad c_{-} \define -\left(v_+({L_0}) - k_\l v_+'({L_0})\right). \label{cpm}
\end{equation}
The solutions of \eqref{v0eqn} can be written as powers: $v_\pm(\zeta) = \zeta^{\theta_{\pm}},$ with $\theta_{\pm}$ defined in \eqref{eq:theta}.
If at the eigenvalue $\Delta_0$, the roots $\theta_\pm$ are real and distinct, %
then the transcendental equation \eqref{determinanteqn} can be written as \eqref{transcendreal}.

If the roots are complex at the eigenvalue $\Delta_0$, then the real-valued solutions of \eqref{v0eqn} are those given in \eqref{vpmcomplex}, 
where $\theta_{r,i}$ are the real and imaginary parts of $\theta_+$. Then, after some algebra, \eqref{determinanteqn} transforms to \eqref{eq:complex-eta}. %
It cannot happen that the two roots are real and equal, $\theta_{+} = \theta_{-}$, since this will contradict our conclusion in Remark \ref{rem:cases} that the NT region is non-degenerate.
\end{proof}

In the zero transaction cost case $\lambda=0$, the no-trade region collapses and ${L_0}=U_0$, which implies from \eqref{eq:pim} and \eqref{eq:pip} that $\pi_+=\pi_-=\pi_M$, the Merton ratio, and that
\begin{equation}
\Delta_0 = \dmax:=\frac{\mu^2}{2\gamma\sigma^2}. \label{dmaxdef}
\end{equation}
For $\lambda>0$ and small enough, we expect $\Delta_0$ to be close (and smaller than) $\dmax$, and so $\pi_{\pm}$ in \eqref{eq:pi} are real. Moreover, we can expect whether we are in the real or complex case to be determined by the discriminant of the quadratic equation \eqref{eq:theta} for $\theta$, namely
\begin{equation}
\theta_{\mbox{\tiny{disc}}}(\Delta_0) = (k-1)^2 - 4k_1\Delta_0, \qquad k:=\frac{\mu}{\half\sigma^2}, \quad k_1:= \frac{-(1-\gamma)}{\half\sigma^2}.%
\end{equation}
This reveals the following cases, as described in \cite[Lemma 3.1]{guasoni}:
\begin{description}
\item[Case I:] If $\gamma<1$, then $k_1<0$ and $\theta_{\mbox{\tiny{disc}}}(\dmax)>0$, and we will be in the case of real $\theta_\pm$ for $\lambda$ small enough.
\item[Case II:] If $\gamma>1$, then $k_1>0$ and
$ \theta_{\mbox{\tiny{disc}}}(\dmax) = \frac{1}{\gamma}k^2 -2k+1,$
and so we will be in the complex case if
$k\in(\gamma-\sqrt{\gamma(\gamma-1)}, \gamma+\sqrt{\gamma(\gamma-1)})$, and in the real case if $k$ is outside that interval.
\end{description}

\begin{remark}
Additionally, it is shown in \cite{GMGS} that the gap function $\lamt$ defined by 
$\Delta_0=\frac{\mu^2-\lamt^2}{2\gamma\sigma^2}$ %
has the following asymptotic approximation as $\lambda\downarrow0$:
\begin{align}
&\lamt = \gamma\sigma^2\left(\frac{3}{4\gamma}\left(\frac{\mu}{\gamma\sigma^2}\right)^2\left(1-\frac{\mu}{\gamma\sigma^2}\right)^\frac13\right)\lam^\frac13+O(\lam^\frac23). \label{eq:lamt}
\end{align}
However, even though the asymptotic approximation \eqref{eq:lamt} is very accurate for small transaction costs $\lam$, we have used the numerical solution of \eqref{transcendreal} or \eqref{eq:complex-eta} in both cases that the roots $\theta_{\pm}$ are real and complex respectively.
\end{remark}

\subsection{Finding $v^{\lam,1}$ and $\delta_1$} \label{sec:var_param}
In the previous section, we detailed the solution to the constant volatility problem, from which $(v^{\lam,0},\delta_0,\l_0,u_0)$ are found by formulas \eqref{v0sol} using the averaged volatility $\sigb$. In the next proposition we give expressions for $v^{\lam,1}$ and $\delta_1$. In preparation, we define the following constants which will be used in the formulas for the complex $\theta_\pm$ case:
\begin{align}
&\Theta \define \left(\begin{matrix} \theta_r & \theta_i\\ -\theta_i & \theta_r \end{matrix}\right), 
\quad \vec{c}\define\left(\begin{matrix} c_+ \\ c_- \end{matrix}\right) ,\qquad \vec{q} \define \left(\begin{matrix} q_+ \\ q_- \end{matrix}\right)\define\left(\Theta^3-\Theta^2\right)\vec{c}, 
\label{eq:Theta-c}\\
&\vec{\hat{c}} \define \left(\begin{matrix} c_+ \\ -c_- \end{matrix}\right), \quad\vec{\check{q}} \define \left(\begin{matrix} q_- \\ q_+ \end{matrix}\right) ,
\quad \vec{\check{c}} \define \left(\begin{matrix} c_- \\ c_+ \end{matrix}\right),
\qquad 
\vec{\tilde q} \define  \left(\begin{matrix} \tilde q_+ \\ \tilde q_- \end{matrix}\right) \define -\frac{V_3}{\half\sigb^2}\vec{q} + \frac{(1-\gamma)\delta_1}{\half\sigb^2}\vec{c} .\label{eq:vec-q} 
\end{align}

\begin{proposition}\label{prop:delta1v1}
If $\theta_{\pm}(\delta_0)$ are real then 
\begin{align}
\delta_1 = \frac{V_3}{1-\gamma}
\frac{L_+c_{+}^2\left(\u_0^{\Delta\theta} - \l_0^{\Delta\theta}\right) - L_-c_{-}^2\left(\u_0^{-\Delta\theta} - \l_0^{-\Delta\theta}\right) + c_+c_-\Delta\theta(L_++L_-)\log\frac{\u_0}{\l_0}}
{c_{+}^2\left(\u_0^{\Delta\theta} - \l_0^{\Delta\theta}\right) - c_{-}^2\left(\u_0^{-\Delta\theta} - \l_0^{-\Delta\theta}\right) + 2c_{+}c_{-}\Delta\theta\log\frac{\u_0}{\l_0}}, ~~
\label{eq:delta1-real}
\end{align}
where $\Delta\theta:=\theta_+ - \theta_-$, and 
$L_{\pm}:=(\theta_{\pm}-1) \theta_{\pm}^2.$

Moreover, $v^{\lam,1}$ is determined up to an additive multiple of $v^{\lam,0}$ by
\begin{align}
v^{\lam,1} &=
C_+\zeta^{\theta_{+}} - \tilde c_{+}\zeta^{\theta_{+}}\log\zeta + \tilde c_{-}\zeta^{\theta_{-}}\log\zeta + \xi\left(c_{+}\zeta^{\theta_{+}} + c_{-}\zeta^{\theta_{-}} \right),\label{eq:v1}
\end{align}
for any $\xi\in\R$, where $C_{+}$ and $\tilde c_{\pm}$ are given by 
\begin{align}
\tilde c_{\pm} &\define-\frac{c_{\pm}\left((1-\gamma)\delta_1 -V_3 L_\pm\right)}{\frac12\sigb^2\Delta\theta}\label{eq:tilde-cpm},\\
C_{+} &\define \frac{\tilde c_{-}\left(\l_0\log \l_0 - k_\l(1+\theta_-\log \l_0)\right)}
{k_\l\theta_{+}{\l_0}^{\Delta\theta} -{\l_0}^{\Delta\theta+1 }}
- \frac{\tilde c_{+}\left(\l_0\log \l_0 - k_\l(1+\theta_+\log \l_0)\right) }{k_\l\theta_{+} -{\l_0}}.~~~\label{eq:C+-real}
\end{align}

Otherwise, if $\theta_{\pm}$ are complex  with real and imaginary parts $(\theta_r,\theta_i)$, then 
\begin{equation}
\delta_1 = \frac{V_3}{1-\gamma}\left(
\frac{\left[\half(\vec{\hat{c}}^T\vec{q})\sin(2\theta_i\eta) + (\vec{c}^T\vec{q})\theta_i\eta - 
\half(\vec{c}^T \vec{\check{q}} )\cos(2\theta_i\eta)\right]_{\eta=\log\l_0}^{\log u_0}}{\left[\half(\vec{\hat{c}}^T\vec{c})\sin(2\theta_i\eta) + (\vec{c}^T\vec{c})\theta_i\eta - 
\half(\vec{c}^T\vec{\check{c}})\cos(2\theta_i\eta)\right]_{\eta=\log\l_0}^{\log u_0}}
\right), \label{delta1complex}
\end{equation}
where we use the definitions in \eqref{eq:Theta-c}-\eqref{eq:vec-q}.
In this case, 
$ v^{\lam,1}$ is determined up to an additive multiple of $v^{\lam,0}$
\begin{align}
 v^{\lam,1} &= A_+(\zeta)v_+(\zeta) + A_-(\zeta)v_-(\zeta) + C_+v_+(\zeta)  + \xi v^{\lam,0}(\zeta),\label{eq:v1-complex}
\end{align}
for any $\xi\in\R$, where $v_\pm$ are given by \eqref{vpmcomplex}, and  
\begin{align}
&C_+ \define -\frac{(1+\l_0)(A_+v_++A_-v_-)'(\l_0) - (1-\gamma)(A_+v_++A_-v_-)(\l_0)}{(1+\l_0)v_+'(\l_0) - (1-\gamma)v_+(\l_0)}, \label{eq:C-plus-complex}\\
&A_{\pm}(\zeta)\define \mp\frac{\tilde q_{\mp}}{2\theta_{i}}\log\zeta +\frac{\tilde q_{\mp}} {4\theta_{i}^2}\sin(2\theta_{i}\log\zeta)\pm \frac{\tilde q_{\pm}} {4\theta_{i}^2}\cos(2\theta_{i}\log\zeta).\label{eq:A-plus-complex}
 \end{align}

\end{proposition}

The proof is given in Appendix \ref{propproof}.

These expressions allow us to compute $\l_1$ and $\u_1$ from \eqref{eq:l1} and \eqref{eq:u1} respectively. Note that the arbitrary multiple of $v^{\lam,0}$ has no influence in these expressions, because of the zero boundary conditions \eqref{eq:boundary1} and \eqref{eq:boundaryu1}.
\section{Slow-scale Asymptotics}\label{sec:buildSolSlow}
We now consider another stochastic volatility approximation, but this time with slow-scale stochastic volatility:
\begin{align}
\frac{dS_t}{S_t} &= (\mu+r)\,dt + f(Z_t)\,dB^1_t\\
dZ_t
		&=		{\eps} \alpha(Z_t) \, dt + {\sqrt{\eps}}\beta( Z_t ) \, dB_t^2 , 
		\label{eq:Z-slow}
\end{align}
where the Brownian motions $(B^1,B^2)$ have correlation structure $d\< B^1, B^2 \>_t=\rho\,dt$.
As described in \cite{fpss-book}, there is empirical evidence for both a fast and slow scale in market volatility. Here we treat the optimal investment with transaction costs problem separately under each scenario for simplicity of exposition. The two approaches can be considered as different approximations for understanding the joint effects of stochastic volatility and costs of trading.

Then the analog of the HJB equation \eqref{eq:HJB-reduced} is
\begin{align}
&\max\left\{ \( \M_0  -\(1-\gamma\)\delta^\eps\cdot \)+{\sqrt{\eps}}\M_1 + \eps\M_2,~\B,~ \S   \right\} \v =0,\label{eq:HJB-reduced-slow}
\end{align}
where 
$$\M_2 = \frac{1}{2}\beta(z)^2 \d_{zz}^2+ \alpha(z) \d_z,  \qquad %
\M_1 = \rho f(z) \beta(z)  \d_{z}D_1, \qquad %
\M_0	=  \frac12 f(z)^2D_2 +  \mu D_1,$$  %
and the operators $D_k$ were defined in \eqref{Dkdef}. 
The definitions of the buy and sell operators $\B,\S$ stay the same as in \eqref{eq:B}-\eqref{eq:S}. %
Similarly, we will work with the free-boundary eigenvalue formulation in Section \ref{freebf}, so that the analog of \eqref{NTeqn} in the no trade region is
\begin{equation}
\left(\(\M_0  -\(1-\gamma\)\delta^\eps\cdot \)+{\sqrt{\eps}}\M_1 + \eps\M_2\right) {\v} =0,\qquad
\zeta\in(\l^\eps(z),u^\eps(z)),\label{NTeqn-slow}
\end{equation}
with boundary conditions \eqref{lbcn1}--\eqref{ubcn2}.

We look for an expansion for the value function
\begin{equation}
\v =  v^{\lam,0} + \sqrt{\eps}\, v^{\lam,1} + \eps  v^{\lam,2} + \cdots,   \label{eq:v-expansion-slow}
\end{equation}
as well as for the free boundaries
\begin{equation*}
\l^\eps = \l_0 + \sqrt{\eps}\,\l_1 + \eps\l_2 + \cdots,\qquad u^\eps = u_0 + \sqrt{\eps}\,u_1 + \eps u_2 + \cdots, %
\end{equation*}
and the optimal excess growth rate
$\delta^\eps = \delta_0 + \sqrt{\eps}\delta_1 + \cdots, $
which are asymptotic as $\eps\downarrow0$.

\subsection{Power expansion inside the NT region}\label{sec:power-expan-slow}
We proceed similarly to Section \ref{sec:power-expan} and analyze \eqref{NTeqn-slow} inside the NT region. 
The terms of order one in \eqref{NTeqn-slow} are
\begin{equation}
\(\M_0  -\(1-\gamma\)\delta_0\cdot \)v^{\lam,0} =0.\label{NTeqn-slow1}
\end{equation}
The operators $(\B,\S,\B',\S')$ in the boundary conditions \eqref{lbcn1}--\eqref{ubcn2} do not depend on $\eps$ and so the expansions in Section \ref{sec:bndry} are the same in the slow case as in the fast. Therefore, for the zeroth order problem, we have \eqref{v0bcnl} and \eqref{v0bcnu}.
This is the constant volatility problem with volatility $\sigma=f(z)$, that is frozen at today's level. Therefore we have
\begin{equation}
v^{\lam,0}(\zeta,z)=V_0(\zeta;f(z)), ~~ \delta_0=\Delta_0(f(z)),  ~~\mbox{and }~~ (\l_0(z),u_0(z))=(L_0(f(z)),U_0(f(z))), \label{v0solslow}
\end{equation}
where $(V_0,\Delta_0,L_0,U_0)$ are the solution described in Section \ref{sec:buildSol1}. To simplify notation, we will typically cease to write the argument $f(z)$, and simply write $(V_0,\Delta_0,L_0,U_0)$.

The $O(\sqrt{\eps})$ terms give
\begin{align}
\M_0 v^{\lam,1} - (1-\gamma)\delta_0v^{\lam,1} + \M_1 v^{\lam,0} - (1-\gamma)\delta_1v^{\lam,0}  &=0,
\end{align}
where as before $\delta_1$ remains to be found. 
Therefore, we get that the equation for $v^{\lam,1}$ is
\begin{align}
\lnt(f(z);\delta_0) v^{\lam,1} = - \M_1 v^{\lam,0} + (1-\gamma)\delta_1v^{\lam,0}. 
\label{v1eqn-slow}
\end{align}
where $\lnt$ was defined in \eqref{lntdef}. As the boundary condition expansions from Section \ref{sec:bndry} are the same for the slow-scale volatility case, we have the boundary conditions \eqref{eq:boundary3} and \eqref{eq:boundary5} for $v^{\lam,1}$. 

\subsection{Computation of  $\delta_1$}

As in Section \ref{sec:var_param}, we set $\w$ to be the solution of the adjoint equation \eqref{weqn} with boundary conditions \eqref{eq:bnd1}, but with $\sigb$ replaced by $f(z)$. As before, Lemma \ref{wlemma} carries through with the new notation and we recall $w(\zeta,z) = \zeta^{k(z)-2}v^{\lam,0}(\zeta,z),$ with $ k(z)\define\frac{\mu}{\half f(z)^2}$ similar to \eqref{kdefs}. Multiplying both sides of \eqref{v1eqn-slow} by $\w$ and integrating from $\l_0$ to $\u_0$
we get %
$$\delta_1(z) = \frac{ \int_{\l_0}^{\u_0}\w \M_1v^{\lam,0} d\zeta}{(1-\gamma)\int_{\l_0}^{\u_0}\w v^{\lam,0}d\zeta} = \frac{ V_1(z)\int_{L_0}^{U_0}\w D_1\d_\sigma V_0\,d\zeta}{(1-\gamma)\int_{L_0}^{U_0}\w V_0\,d\zeta},
$$
where $V_0(\zeta;f(z)), L_0(f(z))$, and $U_0(f(z))$ are the constant volatility solution, and we define \begin{align}
V_1(z) \define \rho f(z) f'(z)\beta(z).
\label{eq:V3-slow}
\end{align}

We observe that we need to compute $\vega\define\d_\sigma V_0$, which depends on the eigenvalue $\Delta_0$, and so we will need the derivative $\dpr\define\d_\sigma \Delta_0$. Computing the ``Vega'' of the constant volatility ``value function'' is difficult to perform analytically because $\sigma$ appears in numerous places in the solution constructed in Section \ref{sec:buildSol1}: the transcendental equation 
\eqref{transcendreal} or \eqref{eq:complex-eta} for $\Delta_0$, the quadratic equations \eqref{eq:theta} and \eqref{eq:pim} for $\theta_\pm$ and $\pi_\pm$ and consequently in $L_0$ and ${U_0}$. Numerically, computing a finite difference approximation is simple, but we will need to do so at many values of $\zeta$ for use in certain integrals in the asymptotic correction in the next section. Therefore it is of interest to relate it to derivatives (or Greeks) of $v^{\lam,0}$ in the $\zeta$ variable, in particular the Gamma $V_0''=\partial^2_{\zeta\zeta}V_0$, which we have found to be amenable to computation in the fast asymptotics.

First we give an expression for $\dpr$ which avoids numerical differentiation of the eigenvalue problem. 
\begin{lemma}\label{dprlemma}
The derivative $\dpr$ is given by the following ratio of integrals:
\begin{equation}
 \dpr = \frac{\sigma\int_{L_0}^{{U_0}}wD_2{V_0}\,d\zeta}{(1-\gamma)\int_{L_0}^{{U_0}}w{V_0}\,d\zeta}. \label{dpreqn}
 \end{equation}
\end{lemma}
\begin{proof}
In the NT region and at the boundaries, we have
\begin{align}
\half\sigma^2\zeta^2{V_0''} + \mu\zeta{V_0}' - (1-\gamma)\Delta_0{V_0} &= 0\label{vzODE}\\
\(1+ L_{0} \){V_0'}(L_0)  - (1-\gamma){V_0}(L_0) &=0,\label{lb1}\\
\(1+ L_{0} \){V_0''}(L_0)  + \gamma {V_0'}(L_0) &=0, \label{lb2}\\
\(\frac{1}{1-\lam} + U_{0} \){V_0'}({U_0})  - (1-\gamma){V_0}({U_0}) &=0,\label{ub1}\\
\(\frac{1}{1-\lam}+ U_{0} \){V_0''}({U_0})  + \gamma {V_0'}({U_0})&=0,\label{ub2}
\end{align}
where $'=\frac{d}{d\zeta}$. %
Differentiating the ODE \eqref{vzODE} with respect to $\sigma$, we find that in the NT region, $\vega=\d_{\sigma} {V_0}$ satisfies
\begin{equation}
 \half\sigma^2\zeta^2\vega'' + \mu\zeta\vega' - (1-\gamma)\Delta_0\vega  = -\sigma D_2{V_0} + (1-\gamma)\dpr{V_0}. \label{vegaeqn}
 \end{equation}
We also have by differentiating \eqref{lb1} with respect to $\sigma$ and using the $C^2$ smooth pasting condition \eqref{lb2} for ${V_0}$ that $\vega$ satisfies the usual homogeneous Neumann boundary condition at $L_0$:
$$\(1+ L_{0} \)\vega'(L_0)  - (1-\gamma)\vega(L_0) =0. $$
Similarly, differentiating \eqref{ub1} with respect to $\sigma$ and using \eqref{ub2} gives:
$$\(\frac{1}{1-\lam} + U_{0} \)\vega'({U_0})  - (1-\gamma)\vega({U_0}) =0. $$

We note that since ${V_0}$ is well-defined, so is $\vega$ by differentiation: we are just using equations it must satisfy to try and shortcut its computation. 
Then a Fredholm solvability condition for \eqref{vegaeqn} determines $\dpr$. 
Multiplying equation \eqref{vegaeqn} by the adjoint function $w = \zeta^{k(z)-2}V_0$, integrating by parts and using the boundary conditions satisfied by the vega $\vega$ yields \eqref{dpreqn}.
\end{proof}

The expression for the Vega $\d_\sigma V_0$ is given in Appendix \ref{sec:explicit} in the case of real $\theta_\pm$. The formula in the complex case is very long and we omit it in this presentation. Writing the Vega in terms of spatial derivatives is related to the classical Vega-Gamma relationship for European option prices (see, for instance the discussion in \cite[Section 1.3.5]{fpss-book}), which can be used to show that portfolios that are long Gamma (convex) and long volatility (positive Vega). In the context of the classical Merton portfolio optimization problem with no transaction costs, an analogous relationship between the derivative of the value function with respect to the Sharpe ratio and the negative of the second derivative with respect to the wealth variable is found in \cite[Lemma 3.1]{fsz}. For infinite horizon problems, as here, it is not so direct because there is no time derivative that allows for a simple explicit solution of equation \eqref{vegaeqn} and its boundary conditions that would give $\vega$ in terms of $D_2V_0$, but nonetheless, a useful expression \eqref{vg} can be found.

\subsection{Computation of $v^{\lam,1}$}
We proceed, as in Section \ref{sec:var_param} to use the variation of parameters \eqref{eq:v1-form} to solve the inhomogeneous equation \eqref{v1eqn-slow} with boundary conditions \eqref{eq:boundary3} and \eqref{eq:boundary5}. We recall that the principal solution $V_0$ is given by formulas \eqref{v0formula} and \eqref{cpm}, where $v_\pm$ are the independent solutions of the the ODE (in $\zeta$) \eqref{v0eqn} with the volatility $\sigma=f(z)$. %
Then we have $v^{\lam,1} (z,\zeta)= A_{+}(z,\zeta)v_{+} (\zeta)+ A_{-}(z,\zeta)v_{-}(\zeta),$ where $A_{\pm}$ solve the same system of equations \eqref{eq:c-cond1} and \eqref{eq:c-cond2}, 
\begin{align}
&A'_{+} v_{+} + A'_{-} v_{-}=0,\label{eq:c-cond1-slow}\\
&A'_{+} v_{+}' + A'_{-} v_{-}'= F(z,\zeta), %
\label{eq:c-cond2-slow}
\end{align}
and $'$ denotes the derivative with respect to $\zeta$. 
The only change is that in this case, 
\begin{align}
F(z,\zeta) \define \frac{-V_1(z)D_1\d_\sigma V_0 + (1-\gamma)\delta_1V_0}{\frac12f(z)^2\zeta^2}.
\label{eq:F}
\end{align}
The solution of the system \eqref{eq:c-cond1-slow}--\eqref{eq:c-cond2-slow} is given by 
\begin{equation}
A_{\pm}(z,\zeta) =  \mp\int \frac{ v_{\mp} }{ v_{-}'v_{+} - v_{+}'v_{-} }  F(z,\zeta) d\zeta + C_{\pm}. \label{Apmsol-slow}
\end{equation}
This determines $v^{\lam,1} $. In the Proposition \ref{prop:slow} that follows, we show how these can be explicitly computed in the case of real $\theta_\pm$.

We also note, that the calculations in Section \ref{sec:bndry} are applicable also to the slow-scale case, and we conclude that the corrections to the boundaries $\l_1(z)$ and $\u_1(z)$ can be determined using 
\eqref{eq:l1} and \eqref{eq:u1}.

\subsection{Explicit Computations of  $(v^{\lam,1},\delta_1)$ in the real $\theta_\pm$ case}\label{sec:slow-explicit}
For the rest this section we will show the computations of $v^{\lam,1},$ and $\delta_1$, 
in the case when the roots $\theta_{\pm}(\delta_0)$ of quadratic \eqref{eq:theta} (with $\sigma=f(z)$) are real.
The complex case can also be calculated analytically, but we did not find the formulas to be enlightening, so we choose to omit the presentation of this calculation.
\begin{proposition}\label{prop:slow}
Recall that %
$\Delta\theta:=\theta_+ - \theta_-$, and that $v^{\lam,0}(\zeta)  = c_{+}v_{+}(\zeta) + c_{-}v_{-}(\zeta)$. 
Let  $\dot c_{\pm}\define \d_{\sigma} c_{\pm},~\dot\theta_{\pm}\define \d_{\sigma} \theta_{\pm}.$
If $\theta_{\pm}$ are real then,
\begin{align}
\delta_1(z) &= \frac{V_1(z)}{(1-\gamma)D}\left\{
c_{+}Q_+R_+ -c_{-} Q_-R_-+\(c_{-}Q_+ + c_{+}Q_-\) \Delta \theta\log\frac{\u_0}{\l_0}
\right.\label{eq:delta1-slow}\\
& + \left(c_{+}^2 \theta_{+}\dot\theta_{+} R_+ \log\frac{\u_0}{\l_0}
-c_{-}^2 \theta_{-}\dot\theta_{-} R_-\right) 
\( \log\frac{\u_0}{\l_0} - \frac{1}{\Delta\theta}   \)\\
&\left.+c_{-}c_{+} \(\theta_{+}\dot\theta_{+} + \theta_{-}\dot\theta_{-}\) \Delta \theta\log^2\frac{\u_0}{\l_0}\right\},
\end{align}
where
$Q_\pm  :=  \(   \theta_{\pm}\dot c_{\pm} + c_{\pm}\dot\theta_{\pm}\)$,
$R_\pm  :=  \left(\u_0^{\pm\Delta\theta} - \l_0^{\pm\Delta\theta}\right)$, and 
$D  :=  c_{+}^2R_+ - c_{-}^2R_- + 2c_{+}c_{-}\Delta\theta\log\frac{\u_0}{\l_0}$.
The constants $\dot c_{\pm}$ and $\dot \theta_{\pm}$ are calculated explicitly in terms of $\dpr$ in Appendix \ref{sec:explicit}.
Then $ v^{\lam,1}$ is determined up to a multiple of $v^{\lam,0}$ :
\begin{align}
v^{\lam,1} &= C_{+}\zeta^{\theta_{+}} - \(\tilde c_{+} - \frac{\tilde d_{+}}{\Delta\theta}\)\zeta^{\theta_{+}}\log\zeta +\(\tilde c_{-} + \frac{\tilde d_{-}}{\Delta\theta}\)\zeta^{\theta_{-}}\log\zeta \label{eq:v1-slow} \\
&- \frac{\tilde d_{+}}{2}\zeta^{\theta_{+}}\log^2\zeta+\frac{\tilde d_{-}}{2}\zeta^{\theta_{-}}\log^2\zeta+ \xi\left(c_{+}\zeta^{\theta_{+}} + c_{-}\zeta^{\theta_{-}} \right),
\end{align}
for any $\xi\in\R$, where 
\begin{align}
\tilde c_{\pm}(z)&\define\frac{2}{f^2(z)}\(\frac{Q_\pm V_1(z)  + (1-\gamma)\delta_1c_{\pm}     }{\Delta\theta }  
\),~%
\tilde d_{\pm}(z)\define\frac{2V_1(z)}{f^2(z)}\frac{\dot\theta_{\pm}c_{\pm}\theta_{\pm}}{(\Delta\theta)},\label{eq:tilde-cd-slow}\\
 C_{+} &\define \frac{b_1}
{k_\l\theta_{+}{\l_0}^{\theta_{+}-1} -{\l_0}^{\theta_{+}}},
\label{eq:C-plus-slow}
\end{align}
and $b_1$ is given below in \eqref{eq:b1-slow}.
\end{proposition}

The proof is given in Appendix \ref{propproof2}.

\section{Analysis of the results}\label{sec:results}
The following are graphs of the buy and the sell boundaries $\l_0$ and $\u_0$ and the long-term growth rate $ \delta_0$ 
with constant volatility and the first order approximation to the buy and the sell boundaries with stochastic volatility $\l_0 + \sqrt{\eps}\l_1, \u_0 + \sqrt{\eps}\u_1$ and to the long-term growth rate $\delta_0 + \sqrt{\eps}\delta_1$. We have four separate cases: slow-scale and fast-scale stochastic volatility, and in each cases, two different sets of graphs that illustrate two additional cases: when the roots $\theta_{\pm}$ of equation \eqref{eq:theta} are real, and when they are complex.

The values used to obtain these graphs are $V_3=-1, \mu=7\%, \sigb=f(z)=f'(z)=0.2, \gamma=2$ for the case $\theta_{\pm}$ are real, and $\mu=5\%, \sigb=f(z)=f'(z)=0.2, \gamma=2$ for the imaginary roots $\theta_{\pm}$. Additionally, we have used $\eps = 10^{-3}$ in case of fast-scale real roots, $\eps = 10^{-4}$ in case of fast-scale complex, $\eps = 10^{-6}$ in case of slow-scale real roots, and $\eps = 10^{-3}$ in case of the complex roots. The reason for such vastly different $\eps$ in all the cases, is the desire to have the $O(\sqrt{\eps})$ approximation to be close to the original boundaries. 

We observe that the effect of stochastic volatility is to shift both the buy and the sell boundaries down. Interestingly, in case of fast-scale stochastic volatility, the shifts to the boundaries $\sqrt{\eps}\l_1$ and  $\sqrt{\eps}\u_1$ do not depend on $Z_t$ the current level of stochastic volatility factor. The intuitive explanation of this observation is the following: we emphasize that this is only approximation to the boundary. As such, even if the current position is $O(\eps)$ away from the boundary, then in the time it takes the wealth ratio $\zeta_t$ to reach the boundary, the volatility factor will have changed by $O(1)$. Hence, even at such close proximity to the boundary, the current volatility factor level is not important. What important is the average of the variance $\sigb^2$. Hence, in case of fast-scale stochastic volatility, only the average level $\sigb$ plays a role. The situation is, of course, not analogous in the slow-scale case, where the current level is extremely important. As intuitively, we can use the same level of volatility factor, for a significant amount of time, with insignificant measurement error. 

Intuitively, the effect of stochastic volatility should reduce (or at least not increase) the long-term growth rate $\delta^\eps$, at least for small correlation $\rho$, the intuition coming from Jensen's inequality. Indeed,  
\begin{align}
\E^{x,y,z}_0\left[ \U\left(X_T+Y_T-\lam Y_T^{+}\right)\right] &= \E^{x,y,z}_0\left[\E\left[ \U\left(X_T+Y_T-\lam Y_T^{+}\right)\big| \mathcal F^{B^2} \right]\right]\\
& \le \E^{x,y,z}_0\left[ \U\left(\E\left[X_T+Y_T-\lam Y_T^{+}\big| \mathcal F^{B^2} \right]\right)\right].
\end{align}
The right hand side, for $\rho=0$, is approximately the total wealth using the average volatility $\sigb$, in case of fast-scale volatility, and the  the total wealth using the initial volatility level $z$ in the slow-scale volatility. Hence, we expect $\delta^\eps\le\delta_0$ in those cases. 

It turns out however that the first order effect in case $\rho=0$ is zero. This is clear from the calculations, specifically \eqref{eq:V3}, \eqref{eq:delta1}, \eqref{Fdef}, \eqref{Apmsol}, that $\delta_1=0$, and $v^{\lam,1} =0$ and thus so are $\l_1$ and $u_1$ from \eqref{eq:l1} and \eqref{eq:u1}, since all the terms are proportional to $\rho$. %
This effect has been observed by \cite{fsz}.

Additionally, if $\rho<0$, as is typically observed in the equity markets, 
the effect should be tighter NT region, and lower boundaries, than in case of constant volatility. The intuition here being that if the stock price goes up, the current volatility should be lower, while keeping the average volatility unchanged. This will cause the investor to start selling earlier. Similarly, when stock price goes down, the current volatility will tend to decrease, this will cause the investor to start buying later. Additionally, since the volatility will be quick to return to it's average value, these changes would not be symmetric, and the width of the NT region will decrease. 
These changes can be observed in the following graphs too, especially in graphs of the boundaries as a function of $\mu$ in the the upper left of the graphs in Figures \ref{fig:fast-real}, \ref{fig:fast-complex}, \ref{fig:slow-real}, \ref{fig:slow-complex}.

We observe, that in all the cases the approximation does not differ much, from the original boundaries, as long as we're sufficiently away from the case when Merton's proportion $\pi_M=1.$ In this case, it is known that the NT region degenerates, as it is optimal to trade only once, and invest all the wealth into the risky stock. In this cases, because the boundaries $\l_0, \u_0$ are positive, and increase to $+\infty$ as the Merton's optimal proportion approaches one. Additionally, both $\l_1$, and $\u_1$ are negative, and they converge to $-\infty$ as the Merton's optimal proportion approaches one. Finally, it appears that the speed of the convergence of $\l_1$, and $\u_1$ is faster than that of $\l_0$ and $\u_0$. {Hence, overall the approximation and the original boundaries cross each other before diverging. Even though, for any fixed $\pi_M<1$ the approximation would be very close to the original boundary, if $\eps>0$ is taken to be small enough.} This is another reason why we have used different $\eps$ between the cases in the fast-scale and the slow-scale stochastic volatility. As the behaviors of the graphs is very different in the slow- and fast-scale cases and in the real and imaginary $\theta_{\pm}$, as the Merton's proportion approaches one.%

The change to the equivalent safe rate from $\delta_0$ to $\delta_0 + \sqrt{\eps}\delta_1$ exhibit more stability. In fact, it is almost a parallel shift down, in case of fast-scale stochastic volatility. However, in certain cases, specifically in fast-scale stochastic volatility, with complex roots $\theta_{\pm}$ we see that the approximation can be greater, then the original equivalent safe rate $\delta_0$. This seems to be the case, because the shift is almost, but not entirely parallel, and with approximation term $\delta_1$ changes signs. For small drift $\mu$ and for large $\sigb$ it is positive, whereas it's negative in other cases. In case of slow-scale stochastic volatility factor, this shift is more pronounced, especially for low initial volatility factor $z$ and for high $\mu$.

\begin{figure}[htb]
 \begin{center}
 \includegraphics[width=0.5\linewidth]{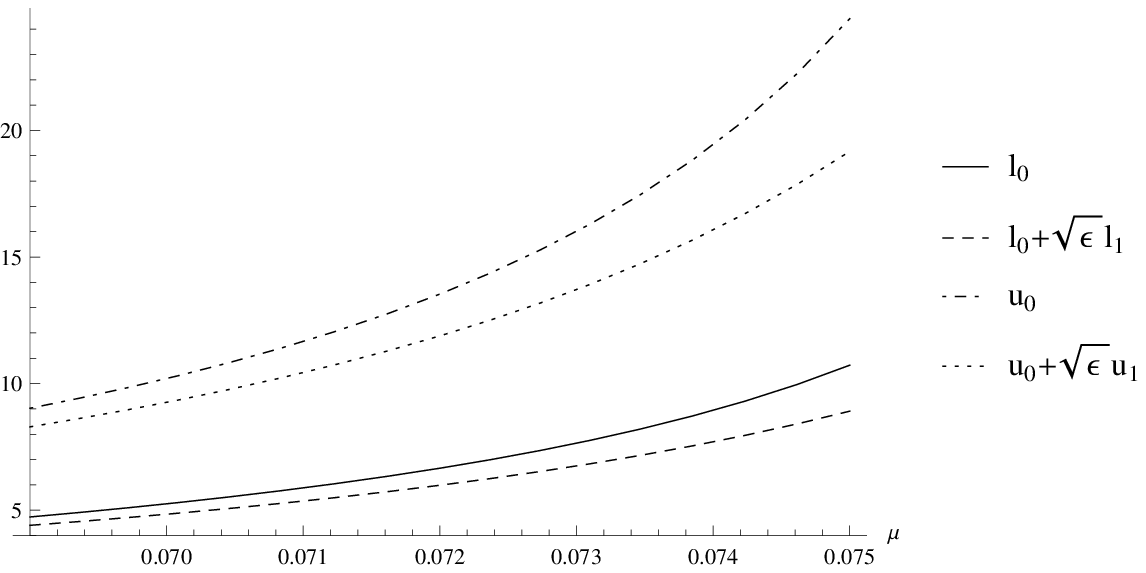}\includegraphics[width=0.5\linewidth]{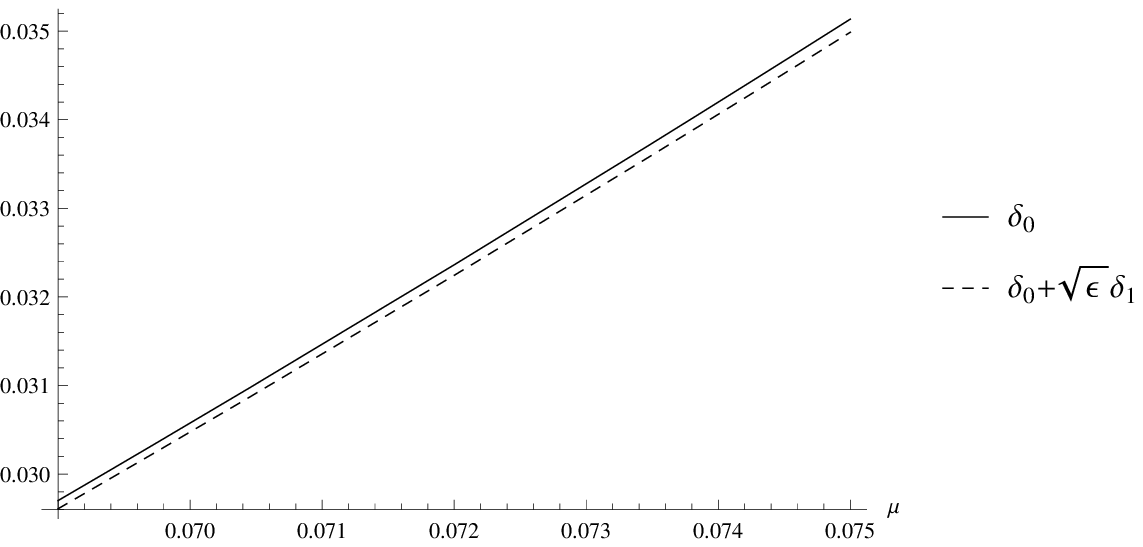}\\%\qquad\qquad
 \includegraphics[width=0.5\linewidth]{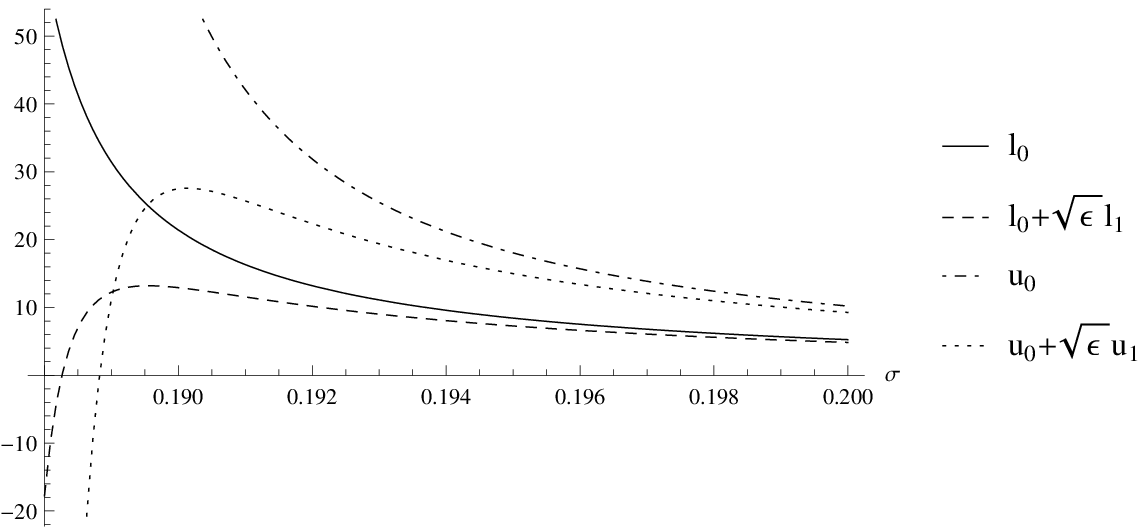}\includegraphics[width=0.5\linewidth]{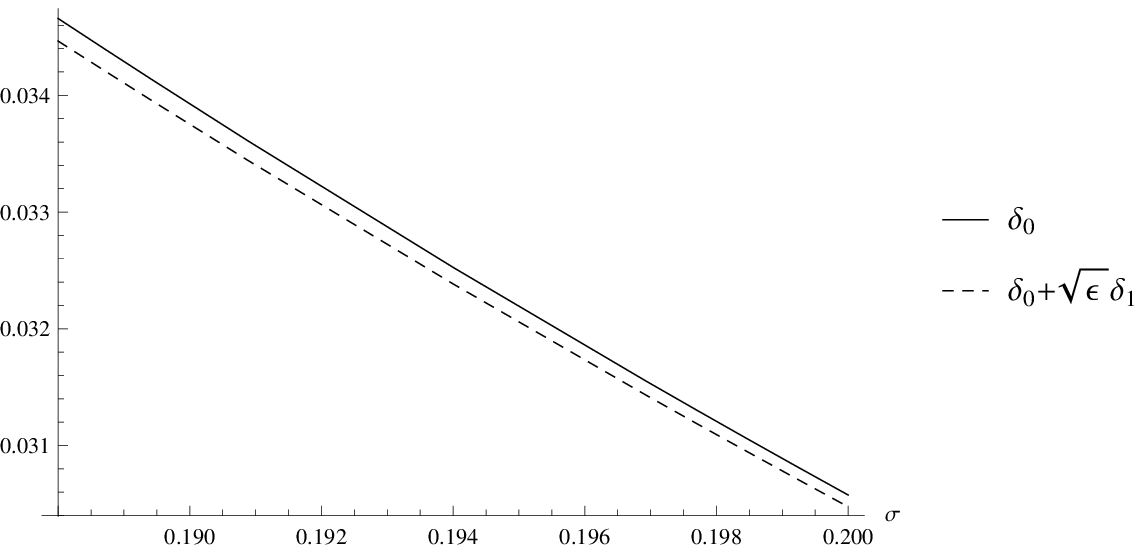}\\%\qquad\qquad
\includegraphics[width=0.5\linewidth]{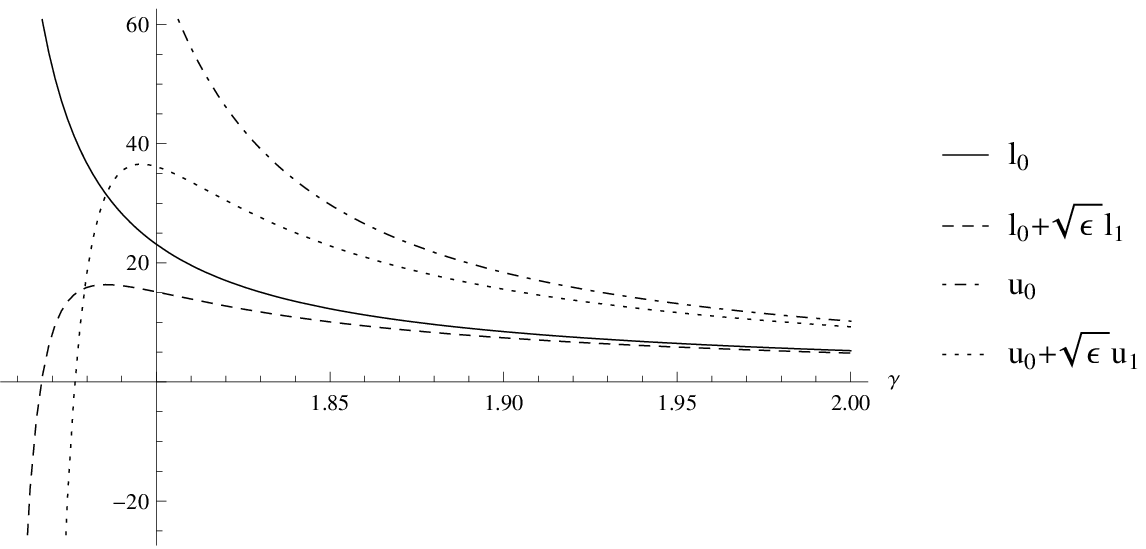}\includegraphics[width=0.5\linewidth]{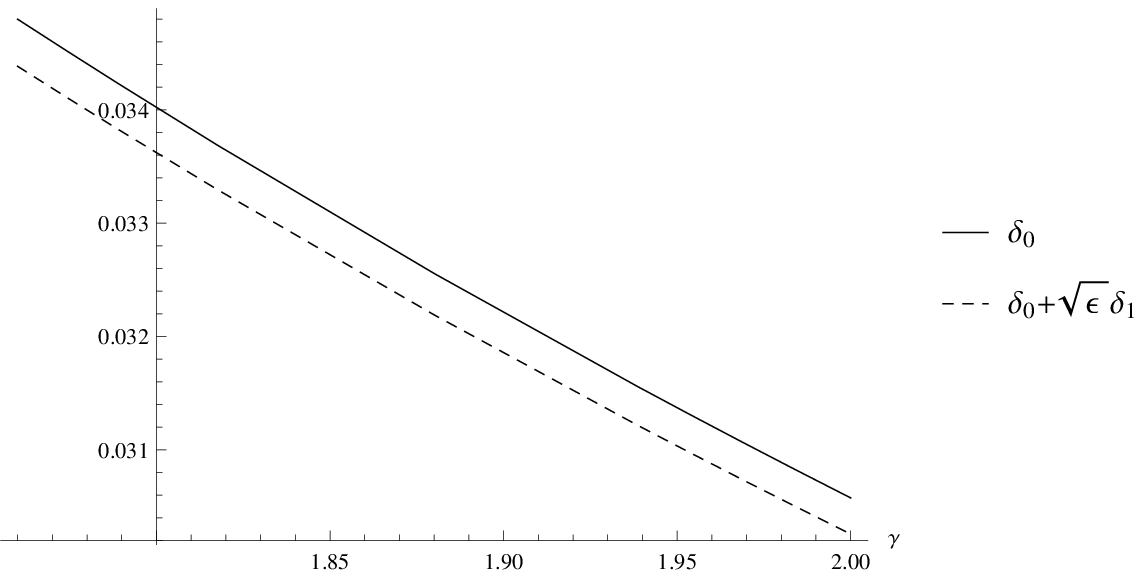}
 \caption{Three graphs of boundaries $\l_0, \u_0$ and $\l_0 + \sqrt{\eps}\l_1, \u_0 + \sqrt{\eps}\u_1$ (left column) and of long-term growth rate $ \delta_0$ and $\delta_0 + \sqrt{\eps}\delta_1$ (right column) in the fast-scale stochastic volatility and in case $\theta_{\pm}$ are real as a function of: Top row$\mu$, middle row
right: $\sigb$, bottom row:  $\gamma$. }
 \label{fig:fast-real}
 \end{center}
 \end{figure}

\begin{figure}[htb]
 \begin{center}
 \includegraphics[width=0.45\linewidth]{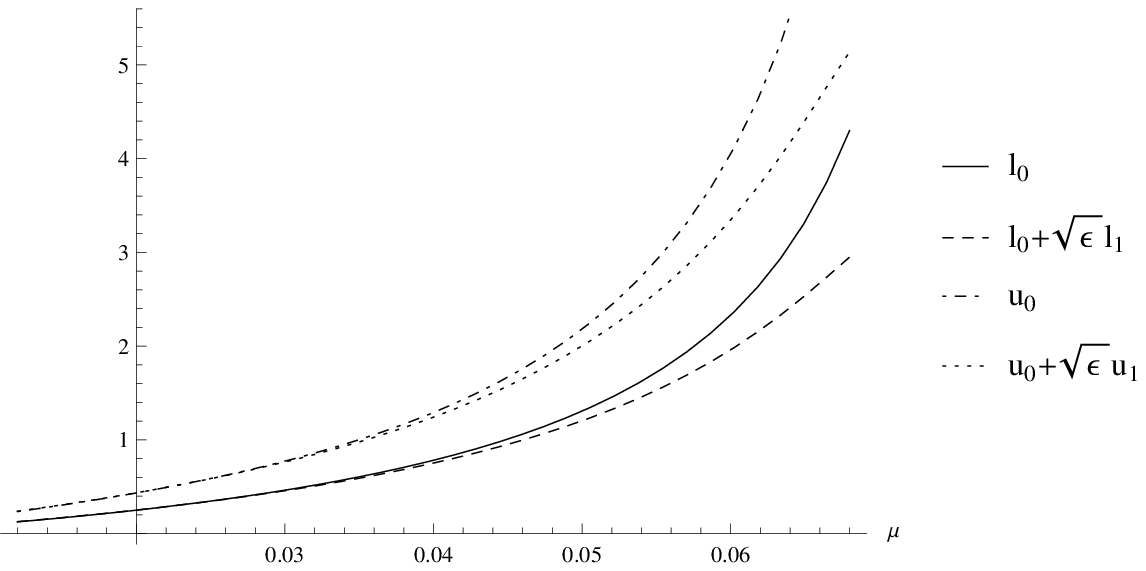}
 \includegraphics[width=0.45\linewidth]{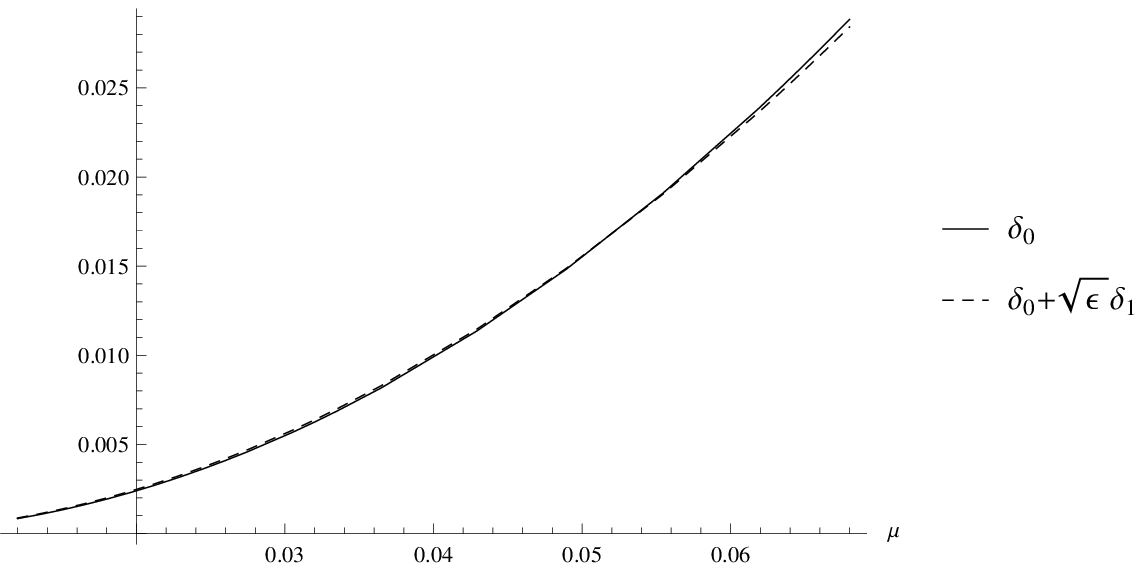}\\
\includegraphics[width=0.45\linewidth]{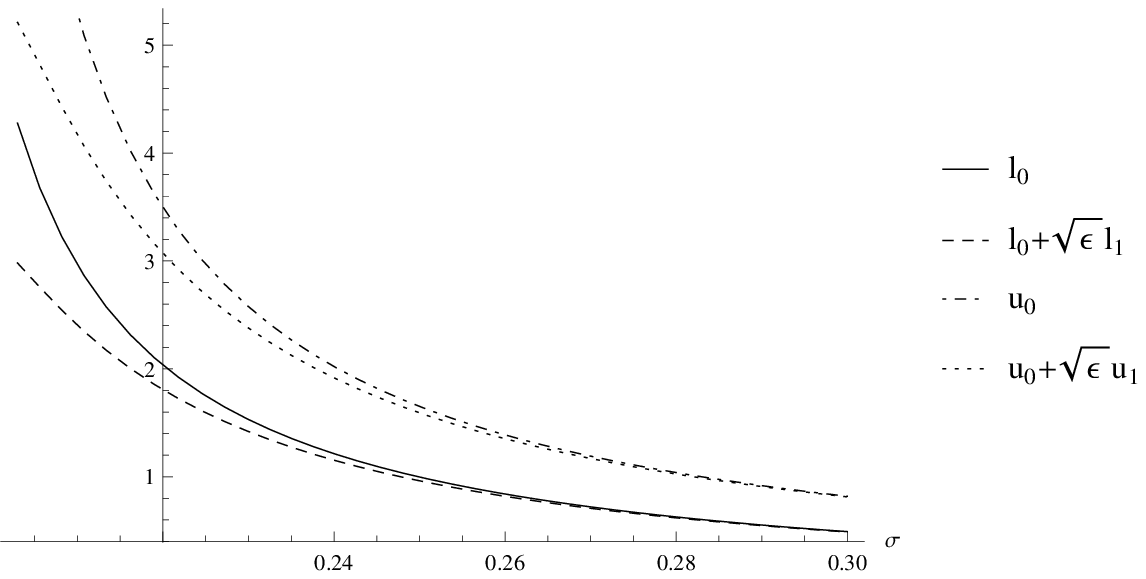}
\includegraphics[width=0.45\linewidth]{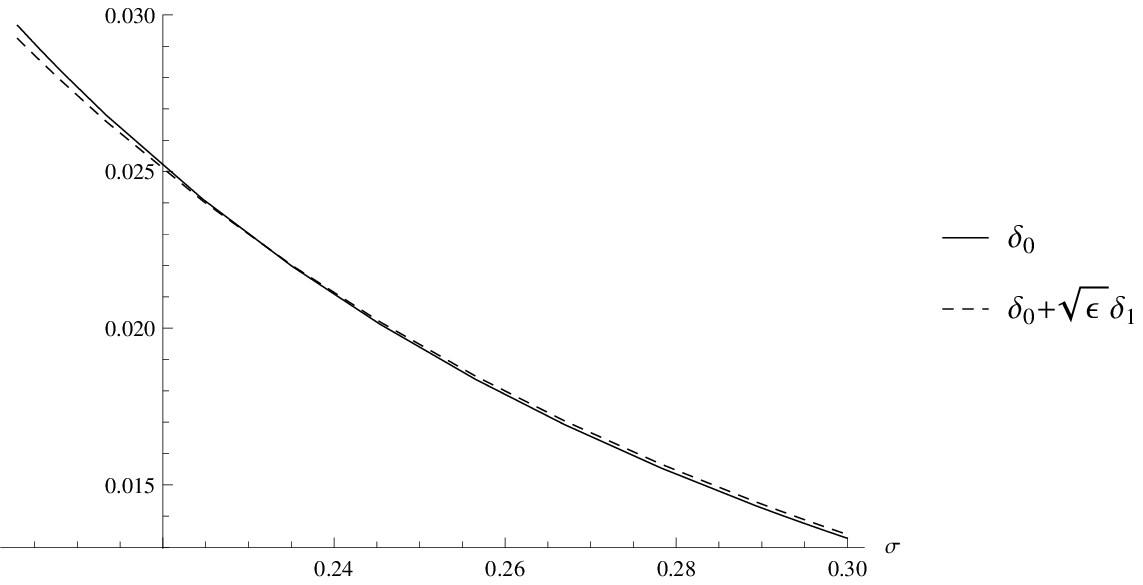}\\
\includegraphics[width=0.45\linewidth]{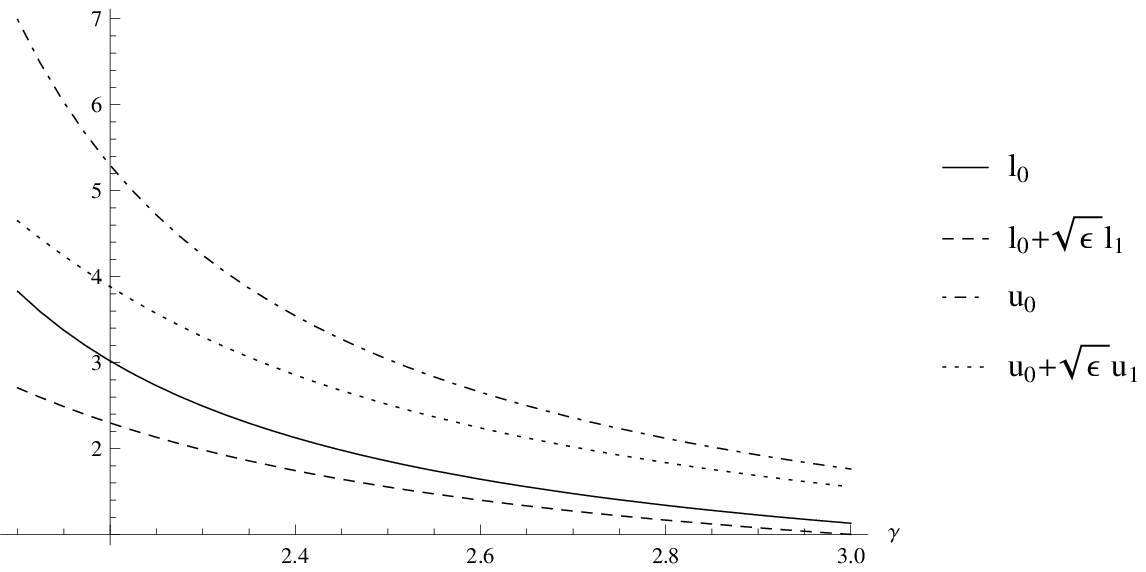}
\includegraphics[width=0.45\linewidth]{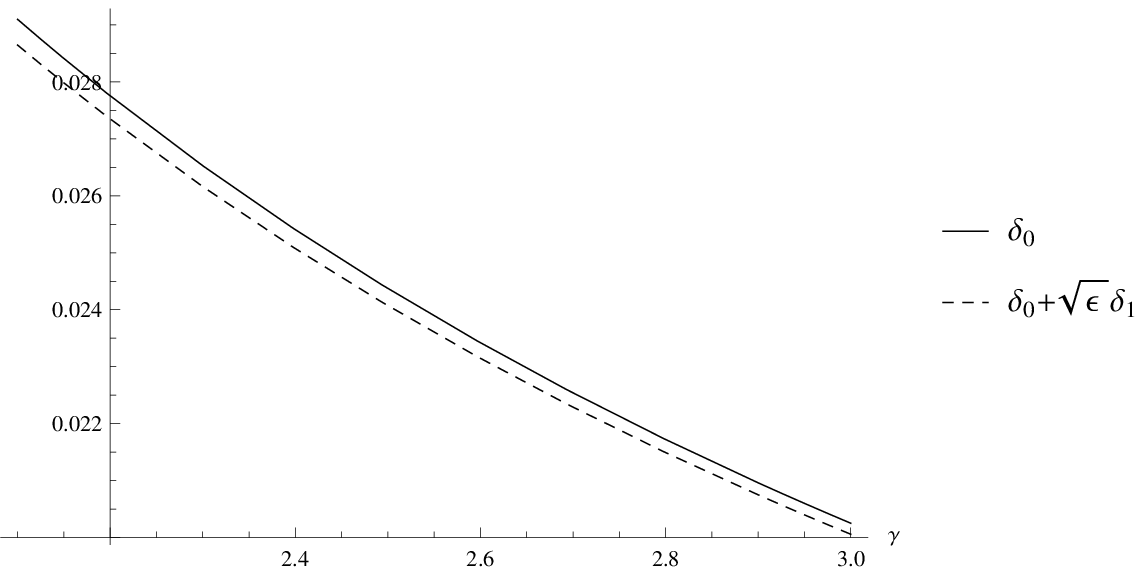}
 \caption{Three graphs of boundaries $\l_0, \u_0$ and $\l_0 + \sqrt{\eps}\l_1, \u_0 + \sqrt{\eps}\u_1$ (left column) and of long-term growth rate $ \delta_0$ and $\delta_0 + \sqrt{\eps}\delta_1$ (right column) in the fast-scale stochastic volatility and in case $\theta_{\pm}$ are complex as a function of: Top row$\mu$, middle row
right: $\sigb$, bottom row:  $\gamma$. }
 \label{fig:fast-complex}
 \end{center}
 \end{figure}

\begin{figure}[htb]
 \begin{center}
 \includegraphics[width=0.45\linewidth]{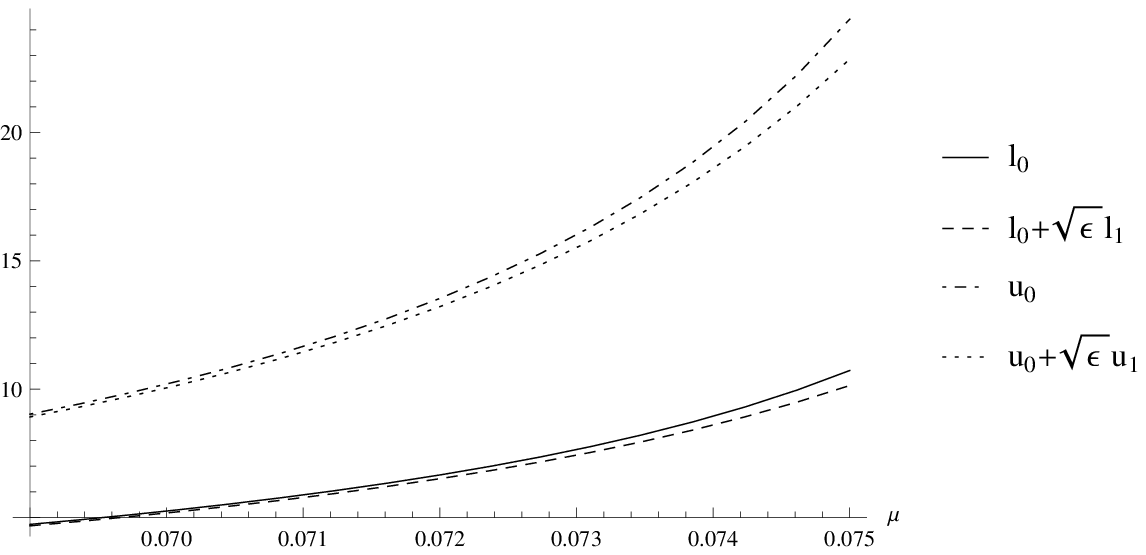}
 \includegraphics[width=0.45\linewidth]{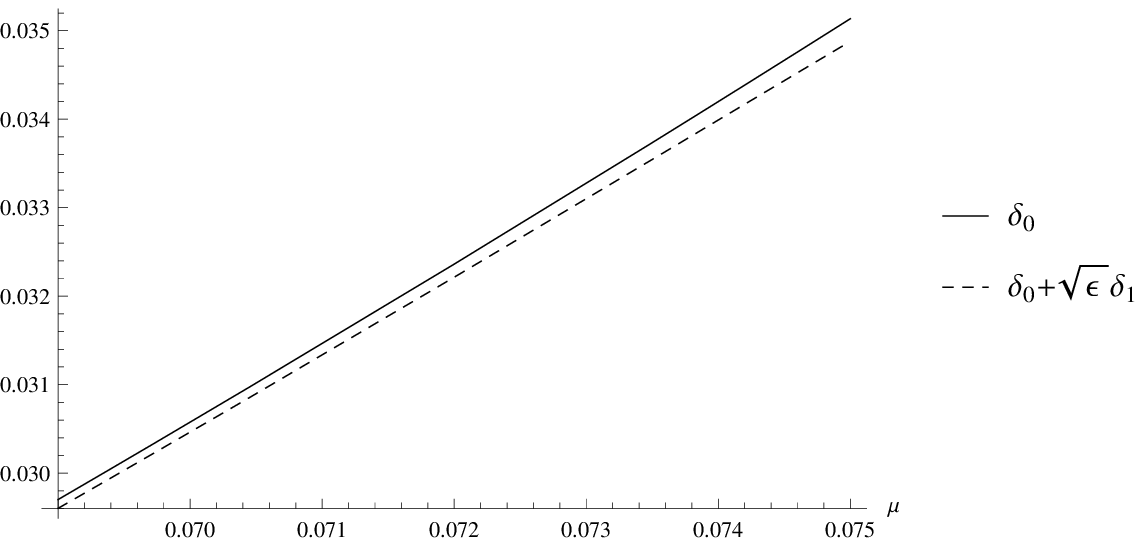}\\%\qquad\qquad
 \includegraphics[width=0.45\linewidth]{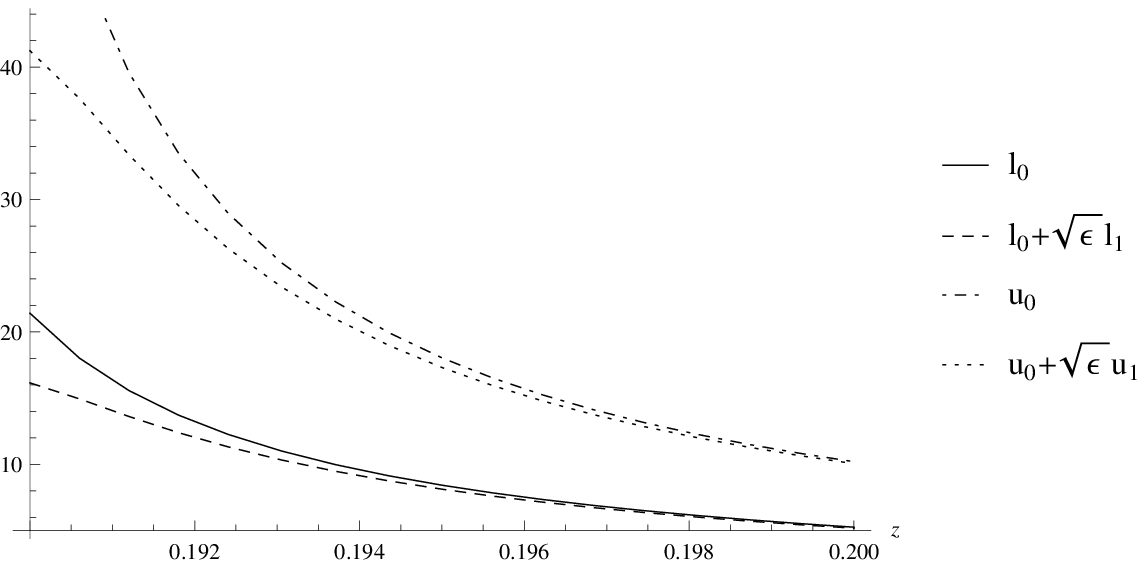}
 \includegraphics[width=0.45\linewidth]{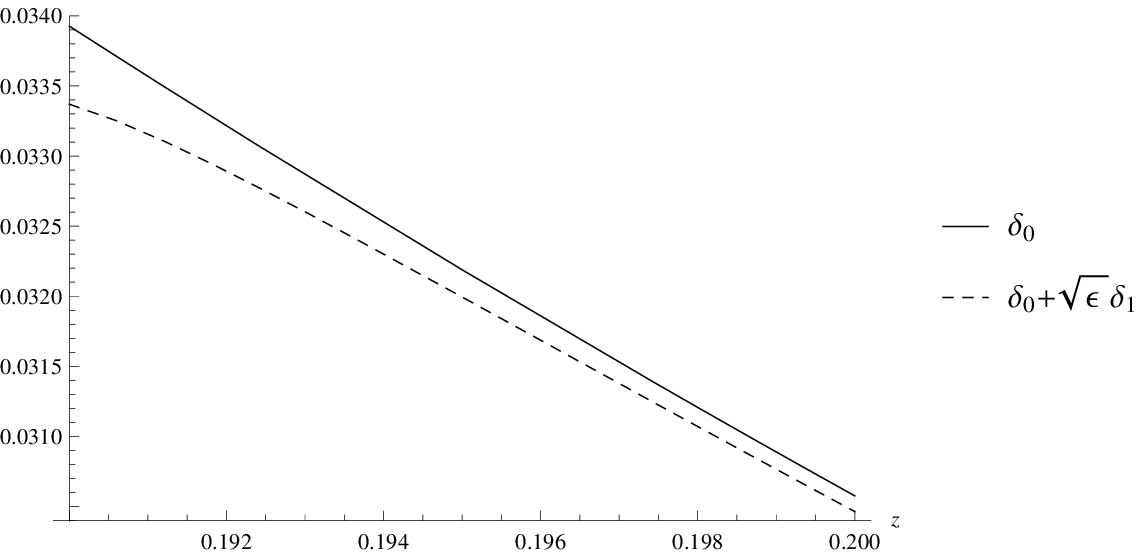}\\%\qquad\qquad
\includegraphics[width=0.45\linewidth]{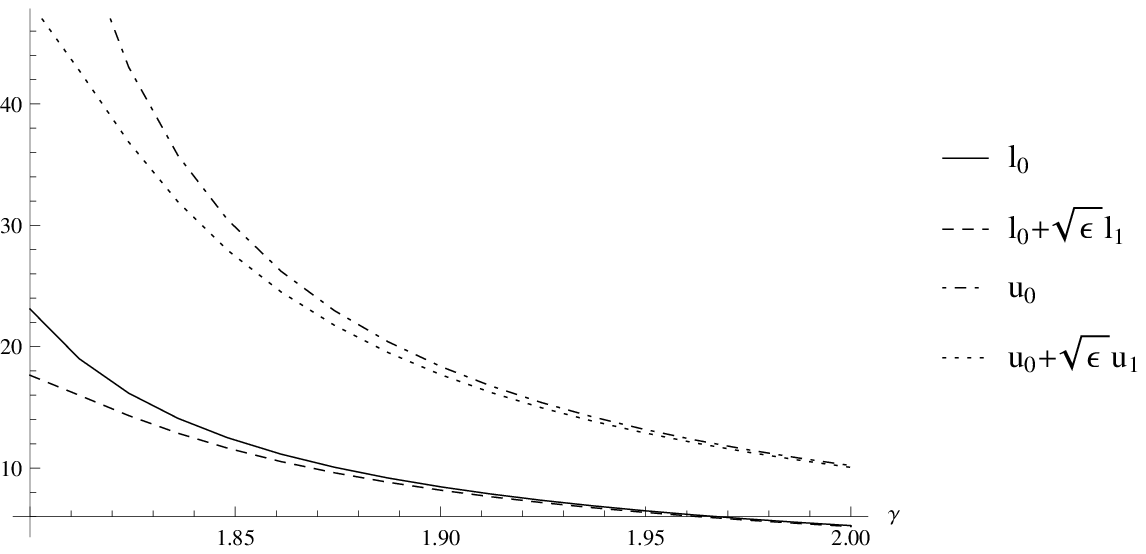}
\includegraphics[width=0.45\linewidth]{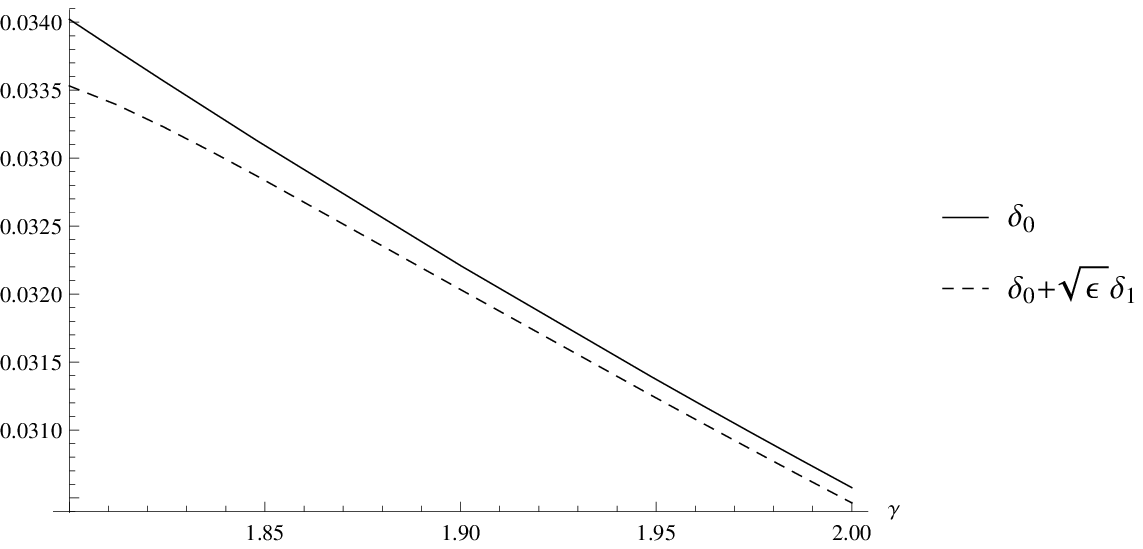}
 \caption{Three graphs of boundaries $\l_0, \u_0$ and $\l_0 + \sqrt{\eps}\l_1, \u_0 + \sqrt{\eps}\u_1$ (left column) and of long-term growth rate $ \delta_0$ and $\delta_0 + \sqrt{\eps}\delta_1$ (right column) in the slow-scale stochastic volatility and in case $\theta_{\pm}$ are real as a function of: Top row$\mu$, middle row
right: $\sigb$, bottom row:  $\gamma$. }
 \label{fig:slow-real}
 \end{center}
 \end{figure}

\begin{figure}[htb]
 \begin{center}
 \includegraphics[width=0.45\linewidth]{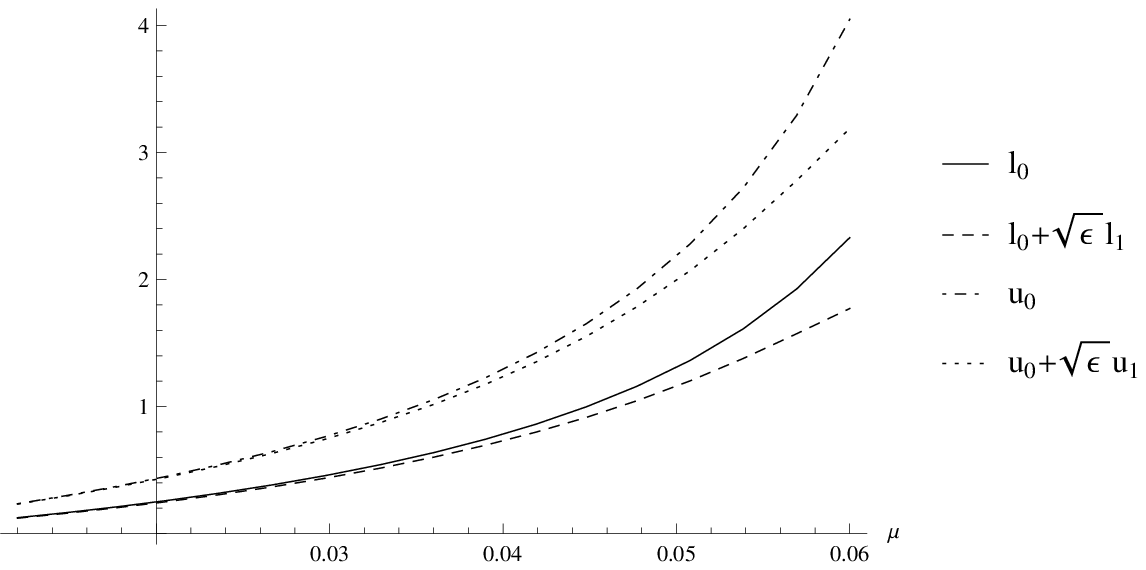}
 \includegraphics[width=0.45\linewidth]{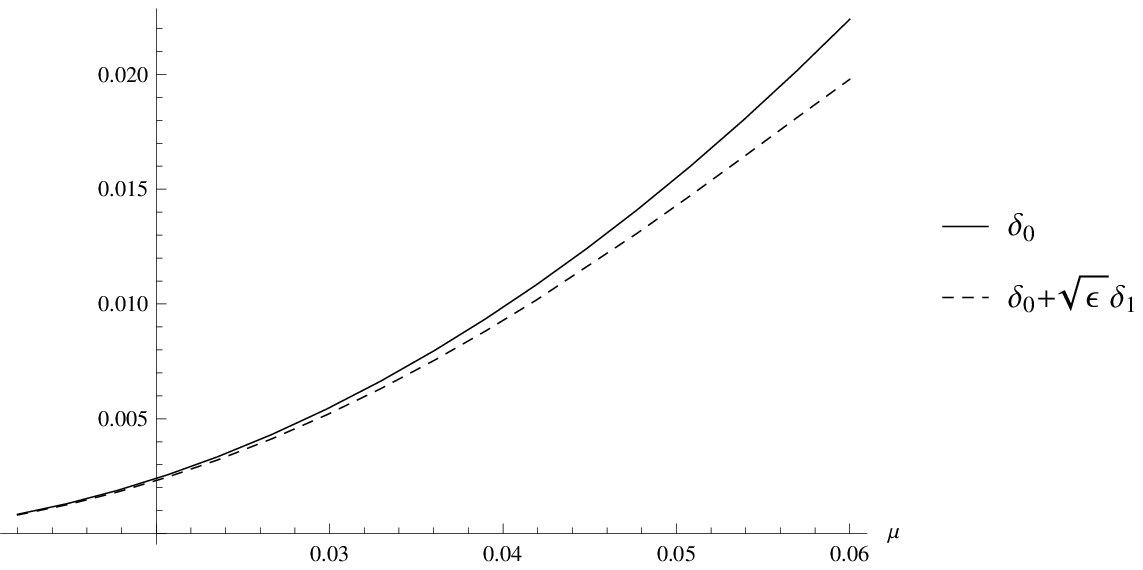}\\
\includegraphics[width=0.45\linewidth]{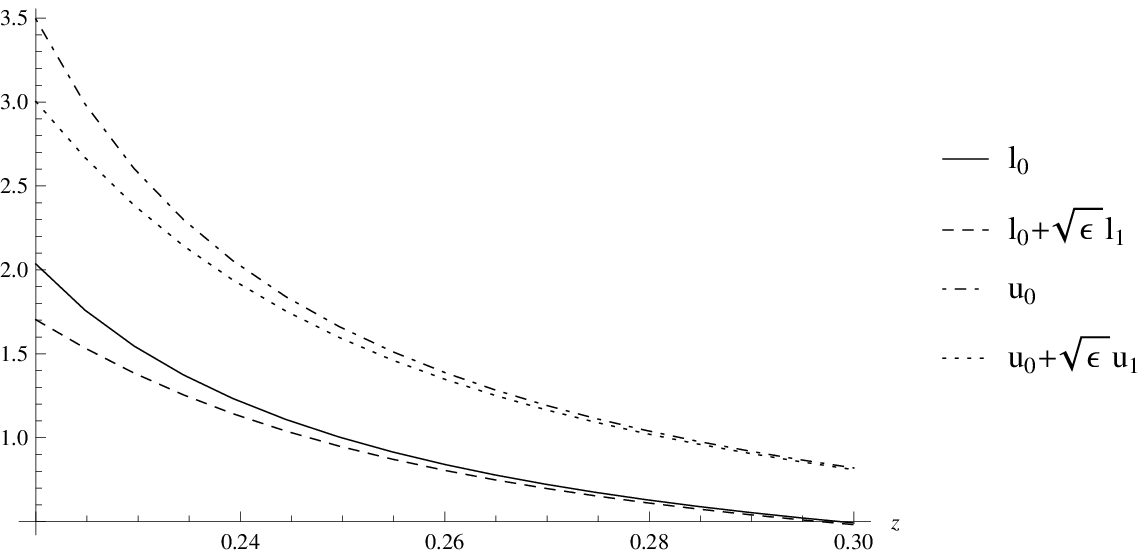}
\includegraphics[width=0.45\linewidth]{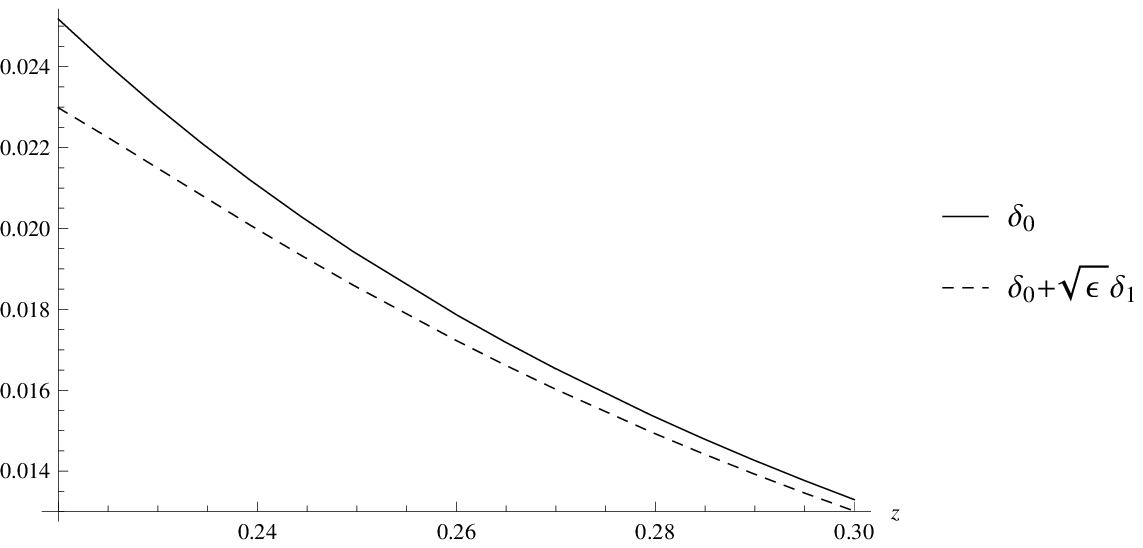}\\
\includegraphics[width=0.45\linewidth]{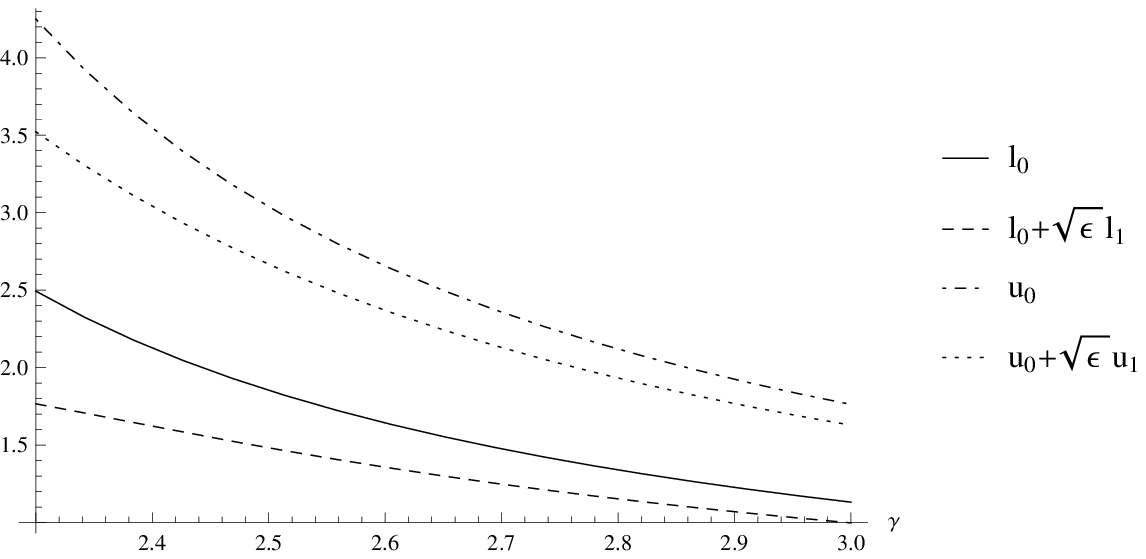}
\includegraphics[width=0.45\linewidth]{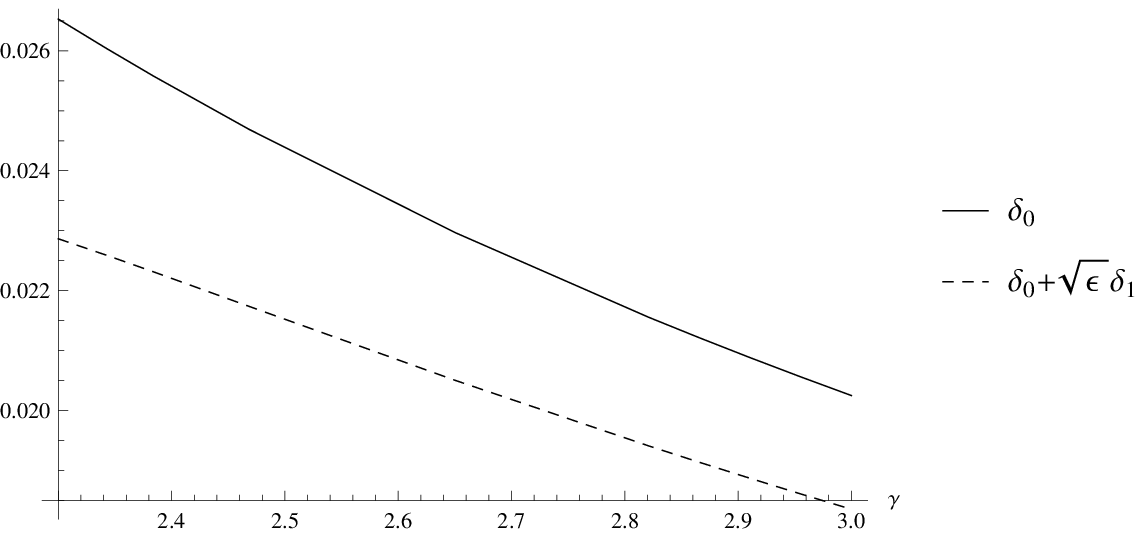}
 \caption{Three graphs of boundaries $\l_0, \u_0$ and $\l_0 + \sqrt{\eps}\l_1, \u_0 + \sqrt{\eps}\u_1$ (left column) and of long-term growth rate $ \delta_0$ and $\delta_0 + \sqrt{\eps}\delta_1$ (right column) in the slow-scale stochastic volatility and in case $\theta_{\pm}$ are complex as a function of: Top row$\mu$, middle row
right: $\sigb$, bottom row:  $\gamma$. }
 \label{fig:slow-complex}
 \end{center}
 \end{figure}

\section{Conclusion}\label{sec:conclusion}
We have analyzed the Merton problem of optimal investment in the presence of transaction costs and stochastic volatility. This is tractable, when the problem is to maximize the long-term growth rate. This leads us to a perturbation analysis of an eigenvalue problem, and shows that the asymptotic method can be used to capture the principle effect of trading fees and volatility uncertainty. In particular we identify that the appropriate averaging, when volatility is fast mean reverting, is given by root-mean-square ergodic average $\sigb$. These techniques can also be adapted to the finite time horizon Merton problem, indifference pricing of options and other utility functions, on a case by case basis.

\appendix

\section{Comparison Theorem}\label{sec:comparisonThm}
In this appendix, we will sketch a proof of the comparison theorem specifically adapted for our case, that shows that $\abs{\V}\le C\abs{V}$. We note, that a standard comparisons theorem for viscosity solutions can be easily adapted for the case $0<\gamma<1$ such as the one in \cite{BichuchShreve}. However, as far as we are aware, this proof is limited to the case $0<\gamma<1$ as it requires the finite values on the boundary of the solvency region $\Sv$ from \eqref{eq:solvency}. 
To circumvent these problems and provide a proof for all cases, we adapt the proof from \cite{JanecekShreve1}. To streamline the proof, we will assume that all local martingales in the following argument are true martingales.

We remind the reader that we have made the assumption that $\hat V,V$ are both smooth functions, namely, 
$ \V,V \in C^{1, 2, 2, 2}\([0,T)\times\Sv\times\R\).$
Additionally, we will also assume that there exists optimal strategies $\(\zhatl, \zhatr\)$ and $\(\zetal, \zetar\)$ for $ \V$, and $V$ respectively. So that for $\V$ we have that
\begin{align*}
\(\(1-\lam\) \d_x - \d_y\)\V(t,x,y,z) &=0 ~ \(t, \frac{y}{x}, z\) \in [0,T) \times [\zhatr(z),\infty)\times \R ,\\ 
(\d_t+\D^\eps)\V &= 0,~  \(t, \frac{y}{x}, z\) \in[0,T) \times [\zhatl(z),\zhatr(z)]\times \R,\\
(\d_y - \d_x)\V&= 0, \(t, \frac{y}{x}, z\) \in[0,T) \times (-\infty,\zhatl(z)]\times \R.
\end{align*}
In other words the no-trade region $\widehat{\mbox{NT}}$ for $\V$ is when the ratio of wealths $\frac{Y_t}{X_t}$ is within $[\zhatl(Z_t), \zhatr(Z_t)]$. Similar, equations hold for $V$ with the no-trade region given by $[\zetal(Z_t), \zetar(Z_t)]$.

{\em Sketch of the proof:} We want to show that $ \abs{\V(t_0, x, y,z)}\le C\abs{V(t_0, x,y,z)}$, for some fixed $(t_0, x,y,z)\in[0,T]\bar\Sv\times\R.$ If $t_0=T$, this follows from the terminal condition, and in case $(x,y)\in\partial \Sv$ it can be shown by adapting the proof in  \cite{ShreveSoner} that the optimal strategy is to liquidate the stock position, resulting in zero total wealth, in which case, $\V=\U$, and $\abs{V}\le C\abs{\U}$. Hence, we proceed with the assumption that $t_0\in[0,T),$ and $(x, y)\in \Sv$ and consider the strategy $\hat X_t, \hat Y_t, \hat Z_t$ starting from $(x, y,z)$ at time $t-$ that keeps, $\frac{\Y_t}{\X_t}$ inside $\widehat{\mbox{NT}}$. For which we have that $\hatL_t = \int_0^t\ind_{\left\{\frac{\Y_s}{\X_s}=\zhatr\right\}} d\hatL_s, \hatM_t = \int_0^t\ind_{\left\{\frac{\Y_s}{\X_s}=\zhatl\right\}} d\hatM_s.$ 
It follows that
\begin{align}
\V(t_0, x, y, z)&  = \V(T, \X_T, \Y_T, \Z_T) -  \int_{t_0}^T  (\d_t+\D^\eps)\V(t,\X_t,\Y_t,\Z_t)dt  \label{eq:Ito-hatV1}\\
&- \int_{t_0}^T f(\Z_t) \d_y  \V(t,\X_t,\Y_t,\Z_t) dB^1_t -\int_{t_0}^T \frac1{\sqrt\eps}\beta\(\Z_t\)  \d_z \V(t,\X_t,\Y_t,\Z_t) dB^2_t  \\
&-\int_{t_0}^T (\d_y - \d_x)\V(t,\X_t,\Y_t,\Z_t)d\hatL_t -\int_{t_0}^T (\d_y - \d_x)\V(t,\X_t,\Y_t,\Z_t)d\hatM_t.
\end{align}
Using the fact that the $dt, d\hatL_t, d\hatM_t$ terms in \eqref{eq:Ito-hatV1} are zero by the optimality of the strategy we conclude that
\begin{align}
\V(t_0, x, y, z)&= \V(T, \X_T, \Y_T, \Z_T) - \int_0^T f(\Z_t) \d_y  \V(t,\X_t,\Y_t,\Z_t) dB^1_t \\
&\quad-\int_0^T \frac1{\sqrt\eps}\beta\(\Z_t\)  \d_z \V(t,\X_t,\Y_t,\Z_t) dB^2_t,
\end{align}
Taking the expectation, we conclude that
$$\V(t_0, x, y, z)  = \E_{t_0}^{x,y,z}\[\V(T, \X_T, \Y_T, \Z_T) \].$$
Writing an equation for $V$ similar to \eqref{eq:Ito-hatV1}, and using the same strategy $\hat X_t, \hat Y_t, \hat Z_t$, we conclude that
\begin{align}
V(t_0, x, y, z)  
&\ge V(T, \X_T, \Y_T, \Z_T) - \int_0^T f(\Z_t) \d_y  V(t,\X_t,\Y_t,\Z_t) dB^1_t \\
&-\int_0^T \frac1{\sqrt\eps}\beta\(\Z_t\)  \d_z V(t,\X_t,\Y_t,\Z_t) dB^2_t,
\end{align}
where we have used the fact that $V$ also solves the HJB equation \eqref{eq:HJB-init}. Again, taking the expectation, and recalling the final time condition  \eqref{eq:final-cond}, we conclude that
$$V(t_0, x, y, z) \ge V(T, \X_T, \Y_T, \Z_T) \ge \frac1C\V(T, \X_T, \Y_T, \Z_T) =\frac1C\V(t_0, x, y, z).$$
The other inequality can be proved similarly, by reversing the roles of $V$ and $\V$. 

\section{Proof of Proposition \ref{prop:delta1v1}\label{propproof}}
To find $v^{\lam,1}$, we use the method of variation of parameters to solve the inhomogeneous equation \eqref{eq:HJB1}, whose source (right-hand side) term  after dividing by the coefficient of the 2nd derivative is
\begin{equation}
 F(\zeta):= \frac{-V_3D_1D_2v^{\lam,0} + (1-\gamma)\delta_1v^{\lam,0}}{\frac12\sigb^2\zeta^2}. \label{Fdef}
 \end{equation}
Specifically, we write 
\begin{align}
v^{\lam,1} (\zeta)= A_{+}(\zeta)v_{+} (\zeta)+ A_{-}(\zeta)v_{-}(\zeta).
\label{eq:v1-form}
\end{align} 
Then we need that $A_{\pm}$ solve the system of equations
\begin{align}
&A_{+}' v_{+} + A_{-}' v_{-}=0,\label{eq:c-cond1}\\
&A_{+}' v_{+}' + A_{-}' v_{-}'= F(\zeta). %
\label{eq:c-cond2}
\end{align}
Indeed, using \eqref{eq:c-cond1}, \eqref{eq:c-cond2}, and the fact that $\lnt v_{\pm} (\zeta) =0$, we see that
\begin{align}
\lnt \left( A_{+}v_{+} + A_{-}v_{-} \right) = \frac12\sigb^2\zeta^2\left(A_{+}' v_{+}' + A_{-}' v_{-}'\right)=-V_3D_1D_2v^{\lam,0} + (1-\gamma)\delta_1v^{\lam,0}.
\end{align}
The solution of the system \eqref{eq:c-cond1}-\eqref{eq:c-cond2} is
\begin{equation}
A_{\pm}(\zeta) =  \mp\int \frac{ v_{\mp} }{ v_{-}'v_{+} - v_{+}'v_{-} }  F(\zeta) d\zeta + C_{\pm}, \label{Apmsol}
\end{equation}
where the constants $C_\pm$ will be determined by the boundary conditions \eqref{eq:boundary3} and \eqref{eq:boundary5}.  
We divide the proof into two cases: the case when $\theta_{\pm}$ the roots of equation \eqref{eq:theta} are real, and the case when they are complex. These are presented in Sections \ref{realcomps} and \ref{sec:fast-complex} respectively.

\subsection{Real $\theta_{\pm}$\label{realcomps}}
When the roots $\theta_\pm$ of the quadratic in \eqref{eq:theta} with volatility $\sigma=\sigb$ and at the eigenvalue $\delta_0$ are real, we have that
$$ w(\zeta) = c_+\zeta^{\theta_++k-2} + c_-\zeta^{\theta_-+k-2}, $$
where $c_{\pm}$ were given in \eqref{cpm}.
We compute $D_1D_2v^{\lam,0}= L_+ c_{+}\zeta^{\theta_{+}} + L_-c_{-}\zeta^{\theta_{-}}$, where $ L_{\pm}$ were defined in Proposition \ref{prop:delta1v1}, and so the calculations for $\delta_1$ from the formula \eqref{eq:delta1} 
leading to \eqref{eq:delta1-real} are straightforward.

Next, we compute that
$$ \frac{v_{\mp}}{ v_{-}'v_{+} - v_{+}'v_{-}}= -\frac{\zeta^{1-\theta_{\pm}}}{\Delta\theta},$$
and so
$ \frac{v_{\mp}}{ v_{-}'v_{+} - v_{+}'v_{-}}F = \tilde c_{\pm}\zeta^{-1} + \tilde c_{\mp}\zeta^{\mp\Delta\theta-1},$ %
where $\tilde c_{\pm}$ are given in \eqref{eq:tilde-cpm}.
Then, from \eqref{Apmsol}, we have
$$ A_{\pm} = \mp\tilde c_{\pm}\log\zeta +  \frac{\tilde c_{\mp}}{\Delta\theta}\zeta^{\mp\Delta\theta} + C_\pm, $$
and, from \eqref{eq:v1-form}, we have
\begin{equation}
v^{\lam,1} = C_+\zeta^{\theta_{+}} + C_-\zeta^{\theta_{-}} - \tilde c_{+}\zeta^{\theta_{+}}\log\zeta + \tilde c_{-}\zeta^{\theta_{-}}\log\zeta , \label{v1formula}
\end{equation}
where we have absorbed some constants into $C_{\pm}$ and retained the same label as they have yet to be determined. 

Inserting \eqref{v1formula} into the boundary conditions  \eqref{eq:boundary3} and \eqref{eq:boundary5} and dividing by $1-\gamma$, we obtain:
\begin{align}
M\left(\begin{array}{c}C_{+}\\ C_{-}\end{array}\right) = b,
\label{eq:system}
\end{align}
where the matrix $M$ from \eqref{mtx1} evaluates in this case to
\begin{align}
M=\left(\begin{array}{cc} k_\l\theta_{+}{\l_0}^{\theta_{+}-1} - {\l_0}^{\theta_{+}}   & k_\l\theta_{-}{\l_0}^{\theta_{-}-1} -{\l_0}^{\theta_{-}} \\
k_u\theta_{+}{\u_0}^{\theta_{+}-1} -{\u_0}^{\theta_{+}}  & k_u\theta_{-}{\u_0}^{\theta_{-}-1} -{\u_0}^{\theta_{-}} 
\end{array}\right),\label{eq:M}
\end{align}
and the vector $b$ is
\begin{align}
b= &
 \left(\begin{array}{c}
 \tilde c_{-}{\l_0}^{\theta_{-}-1}\left(\l_0\log \l_0 - k_\l(1+\theta_-\log \l_0) \right)\\
 \tilde c_{-}{u_0}^{\theta_{-}-1}\left(\u_0\log u_0 - k_u(1+\theta_-\log u_0)\right)
\end{array}\right)\\
&-
 \left(\begin{array}{c}
\tilde c_{+}{\l_0}^{\theta_{+}-1}\left(\l_0\log \l_0 - k_\l(1+\theta_+\log \l_0)\right)\\
  \tilde c_{+}{u_0}^{\theta_{+}-1}\left(\u_0\log u_0 - k_u(1+\theta_+\log u_0)\right)
 \end{array}\right)
\end{align}
and $(k_\l,k_u)$ were defined in \eqref{klku} and we insert the replacements $(L_0,U_0)=(\l_0,u_0)$.

We recall that $M$ is a singular matrix, as we have required that its determinant is zero by choice of $\delta_0$ in \eqref{determinanteqn}. The Fredholm alternative solvability condition for $b$ is satisfied by choice of $\delta_1$ in  \eqref{eq:delta1}. Thus, we get a particular solution by taking $C_{-}=0$ and $C_{+}$ as given by \eqref{eq:C+-real}.
This determines $v^{\lam,1}$ as given in \eqref{eq:v1} up to the addition of a multiple of $v^{\lam,0}$.

\subsection{Complex $\theta_{\pm}$}\label{sec:fast-complex}
When the roots $\theta_\pm$ of the quadratic in \eqref{eq:theta} at the eigenvalue $\delta_0$ are complex, we have
\begin{equation*}
v^{\lam,0}=
c_+v_+(\zeta) + c_-v_-(\zeta) = 
\zeta^{\theta_r}(c_+\cos(\theta_i\log\zeta)+c_-\sin(\theta_i\log\zeta)),
\end{equation*}
where $\theta_r=-\half(k-1)$ using the notation for $k=\frac{\mu}{\frac12\sigb^2}$ defined in \eqref{kdefs}, and c$_{\pm}$ were chosen in \eqref{cpm}. 

We first compute $\delta_1$. From \eqref{eq:delta1}, we have $\delta_1 = \frac{V_3}{1-\gamma}(I_1/I_2)$, where $I_{1,2}$ are the integrals in the numerator and denominator respectively to be computed. Using the change of variable $\eta=\log\zeta$, we have
$$ I_1 = \int_{\l_0}^{\u_0}\w D_1D_2v^{\lam,0} d\zeta= \int_{\log\l_0}^{\log u_0} 
e^{((k-2)+\theta_r+1)\eta} \vec{c}^T\vec{T}(\eta)\left(\frac{\pa^3}{\pa\eta^3} - \frac{\pa^2}{\pa\eta^2}\right)e^{\theta_r\eta}\vec{c}^T\vec{T}(\eta)\,d\eta,
$$
where we define $\vec{T}(\eta) \define \left(\begin{matrix} \cos(\theta_i\eta) \\ \sin(\theta_i\eta) \end{matrix}\right), $ and $\vec{c}=\left(\begin{matrix} c_+ \\ c_- \end{matrix}\right)$ as was defined in \eqref{eq:Theta-c}.
Differentiating the formula in \eqref{vpmcomplex} amounts to multiplying the coefficients of the $\cos$ and $\sin$ terms by the matrix $\Theta$. 
Then $I_1$ reduces to
\begin{eqnarray*} 
I_1 &=& \int_{\log\l_0}^{\log\u_0}\vec{c}^T\vec{T}(\eta)\,\vec{q}^T\vec{T}(\eta)\,d\eta\\
&=& \frac{1}{2\theta_i}\left[\half(\vec{\hat{c}}^T\vec{q})\sin(2\theta_i\eta) + (\vec{c}^T\vec{q})\theta_i\eta - 
\half(\vec{c}^T \vec{\check{q}} )\cos(2\theta_i\eta)\right]_{\log\l_0}^{\log u_0},
\end{eqnarray*}
where $\vec{\hat{c}}, \vec{\check{q}}$ were given in \eqref{eq:vec-q}.
Similarly, we have 
$$ I_2 = {\int_{\l_0}^{\u_0}\w v^{\lam,0}d\zeta}%
= \frac{1}{2\theta_i}\left[\half(\vec{\hat{c}}^T\vec{c})\sin(2\theta_i\eta) + (\vec{c}^T\vec{c})\theta_i\eta - 
\half(\vec{c}^T\vec{\check{c}})\cos(2\theta_i\eta)\right]_{\log\l_0}^{\log u_0}, $$
where $\vec{\check{c}}$ was also given in \eqref{eq:vec-q}. These lead to the expression \eqref{delta1complex} for $\delta_1$.

Next, we compute $A_{\pm}$ in \eqref{Apmsol} where $F$ was defined in \eqref{Fdef}. Again in the variable $\eta=\log\zeta$, we have
$$\tilde A_{\pm}(\eta) = \mp\int\frac{\tilde{v}_\mp(\eta)}{e^{-\eta}\tilde W(\eta)}\tilde{F}(\eta)e^\eta\,d\eta + C_\pm, $$
where $\tilde v_\pm$ denotes $v_\pm$ in $\eta$ co-ordinates, and the Wronskian simplifies to $\tilde W(\eta)=\theta_i e^{2\theta_r\eta}$.
We find that $\tilde{F}(\eta) = e^{-2\eta+\theta_r\eta}\vec{\tilde q}^T\vec{T}(\eta), $ with $\vec{\tilde q} $ was defined in \eqref{eq:vec-q}.
Then %
we obtain
\begin{eqnarray*}
\tilde A_{+} (\eta)&=& -\frac{\tilde q_{-}}{2\theta_{i}}\eta +\frac{\tilde q_{-}} {4\theta_{i}^2}\sin(2\theta_{i}\eta)+ \frac{\tilde q_{+}} {4\theta_{i}^2}\cos(2\theta_{i}\eta)+ C_+\\
\tilde A_{-} (\eta)&=&+\frac{\tilde q_{+}}{2\theta_{i}}\eta +\frac{\tilde q_{+}} {4\theta_{i}^2}\sin(2\theta_{i}\eta)- \frac{\tilde q_{-}} {4\theta_{i}^2}\cos(2\theta_{i}\eta) + C_-.
\end{eqnarray*}
The constants $C_\pm$ are determined by the boundary conditions  \eqref{eq:boundary3} and \eqref{eq:boundary5}. As before, we can take $C_-=0$. Therefore, from \eqref{eq:v1-form}, we have $v^{\lam,1} $ is given by \eqref {eq:v1-complex} using definitions of $C_{+}$ and $A_{\pm}$ in \eqref{eq:C-plus-complex}--\eqref{eq:A-plus-complex}.

\section{Explicit Calculation of the Vega in the Real Case\label{sec:explicit}}%
We will use \eqref{cpm} to calculate $\dot c_{\pm}$ from equation \eqref{eq:v0-sigma}. First, using \eqref{eq:pim} and \eqref{eq:pip}, we find that
\begin{align*}
\dsig\pi_{\pm} = \frac{\dpr + \gamma\sigma\pi_{\pm}^2}{\mu-\gamma\sigma^2\pi_{\pm}},
\end{align*}
where $\dpr:=\dsig\Delta_0$ is computed in Lemma \ref{dprlemma}.
Furthermore, from the same equations we also have
\begin{align*}
\dsig{L_0} &= (1+L_0)^2 \dsig\pi_{-}  = (1+L_0)^2 \frac{\dpr + \gamma\sigma\pi_{-}^2}{\mu-\gamma\sigma^2\pi_{-} },\\
\dsig{U_0} &= \(\frac1{1-\lam}+{U_0}\)^2 \dsig\pi_{+}  =  \(\frac{1}{1-\lam}+{U_0}\)^2 \frac{\dpr + \gamma\sigma\pi_{+}^2}{\mu-\gamma\sig^2\pi_{+}}.
\end{align*}
From \eqref{klku} we obtain that
\begin{align*}
\dsig k_\u = \frac{\dsig{U_0}}{1-\gamma},\qquad \dsig k_\l =\frac{\dsig{L_0}}{1-\gamma},
\end{align*}
and from \eqref{eq:theta} we have
\begin{align*}
\dot\theta_{\pm} = \frac{\sigma\theta_{\pm}(1-\theta_{\pm}) + (1-\gamma)\dpr}{\sigma^2\theta_{\pm}+ \(\mu-\frac{\sigma^2}{2}\)},
\end{align*}
It follows from \eqref{cpm} that
\begin{eqnarray*}
\dot {c}_{\pm} &=& \pm v_{\mp}'(L_0)\dsig{L_0}  \pm  v_{\mp}(L_0)\dot\theta_{\mp}\log{L_0} \mp \frac{\dsig{L_0}}{1-\gamma} v'_{\mp}(L_0) \\
&&\mp k_{\l}\(  \dot\theta_{\mp}\frac{v_{\mp}(L_0)}{L_0} +v_{\mp}''(L_0)\dsig{L_0} +  v_{\mp}'(L_0)\dot\theta_{\mp}\log{L_0}\).
\end{eqnarray*}
These are used to compute 
\begin{align}
\d_\sig V_0 = \dot c_{+}v_{+}(\zeta) + \dot\theta_{+} c_{+} v_{+}(\zeta)\log\zeta + \dot c_{-}v_{-}(\zeta) + \dot\theta_{-} c_{-} v_{-}(\zeta)\log\zeta. \label{vg}
\end{align}

\section{Proof of Proposition \ref{prop:slow}\label{propproof2}}
In the case when the roots $\theta_{\pm}$ are real, we have 
\begin{align}
\d_\sig V_0 = \dot c_{+} \zeta^{\theta_{+}} + c_{+}\dot\theta_{+}\zeta^{\theta_{+}} \log\zeta+
\dot c_{-} \zeta^{\theta_{-}} + c_{-}\dot\theta_{-}\zeta^{\theta_{-}} \log\zeta.\label{eq:v0-sigma}
\end{align}
It follows that
\begin{align*}
&D_1 \d_\sig V_0 = Q_+\zeta^{\theta_{+}} 
+ c_{+}\theta_{+} \dot\theta_{+}\zeta^{\theta_{+}} \log\zeta
+ Q_-\zeta^{\theta_{-}} 
+ c_{-}\theta_{-} \dot\theta_{-}\zeta^{\theta_{-}} \log\zeta,\\
&F(z,\zeta) =\Delta\theta\(\tilde c_{+}(z)\zeta^{\theta_{+}-2 } + \tilde d_{+}(z) \zeta^{\theta_{+}-2}\log\zeta
 +\tilde c_{-}(z)\zeta^{\theta_{-}-2 } + \tilde d_{-}(z) \zeta^{\theta_{-}-2}\log\zeta   \),\\
&\frac{v_{\mp}F}{ v_{-}'v_{+} - v_{+}'v_{-}}
= \frac{\tilde c_{\pm}(z) + \tilde d_{\pm}(z) \log\zeta
 +\tilde c_{\mp}(z)\zeta^{\theta_{\mp}-\theta_{\pm} } + \tilde d_{\mp}(z) \zeta^{\theta_{\mp}-\theta_{\pm}}\log\zeta}{\zeta}   ,
\end{align*}
where we have used  that
$ \frac{v_{\mp}}{ v_{-}'v_{+} - v_{+}'v_{-}}= \frac{\zeta^{1-\theta_{\pm}}}{\Delta\theta},$ together with the definitions of $\tilde c_{\pm}$ and $\tilde d_{\pm}$ in \eqref{eq:tilde-cd-slow}.
A calculation shows that 
\begin{align*}
A_{\pm} = \mp\(\tilde c_{\pm}\log\zeta + \frac{\tilde d_{\pm}}{2}\log^2\zeta \mp \frac{\tilde c_{\mp}}{\Delta\theta}\zeta^{\mp\Delta\theta} \mp \frac{\tilde d_{\mp}}{\Delta\theta} \(\log\zeta \pm \frac1{\Delta\theta}\)\zeta^{\mp\Delta\theta} \) + C_{\pm}.
\end{align*}
Therefore, from \eqref{eq:v1-form}, we have
\begin{align}
v^{\lam,1} &= - \(\tilde c_{+} - \frac{\tilde d_{+}}{\Delta\theta}\)\zeta^{\theta_{+}}\log\zeta +\(\tilde c_{-} + \frac{\tilde d_{-}}{\Delta\theta}\)\zeta^{\theta_{-}}\log\zeta - \frac{\tilde d_{+}}{2}\zeta^{\theta_{+}}\log^2\zeta \label{v1formula-slow}\\
&+\frac{\tilde d_{-}}{2}\zeta^{\theta_{-}}\log^2\zeta+C_+\zeta^{\theta_{+}} + C_-\zeta^{\theta_{-}},\nonumber
\end{align}
where we have absorbed some constants into $C_{\pm}$ and retained the same label as they have yet to be determined. 

As before, we obtain a system of equation similar to \eqref{eq:system} for $C_{\pm}$, with the same matrix $M$ defined as before in \eqref{eq:M}, but with different right hand side vector $b
= \left(\begin{array}{c} b_1\\b_2\end{array}\right)$, given by
\begin{align}
b_1=& \left(\log \l_0 - \frac{k_\l}{\l_0}(1+\theta_-\log \l_0)\right)\left[\(\tilde c_{-} + \frac{\tilde d_{-}}{\Delta\theta}\){\l_0}^{\theta_{-}} - \(\tilde c_{+} - \frac{\tilde d_{+}}{\Delta\theta}\){\l_0}^{\theta_{+}}\right]\\
&-(h_+(\l_0)-h_-(\l_0))\log\l_0,
\label{eq:b1-slow}%
\end{align}
where $h_\pm(\l_0)=\frac{\tilde d_{\pm}}{2}{\l_0}^{\theta_{\pm}}\left(\log \l_0 - \frac{k_\l}{\l_0}(2+\theta_\pm\log \l_0)\right)$.
The second component $b_2$ is given by the same formula with $\l_0$ replaced by $\u_0$.
We recall that $M$ is a singular matrix, as we have required that its determinant is zero by choice of $\delta_0$ in \eqref{determinanteqn}. The Fredholm alternative solvability condition for $b$ is satisfied by choice of $\delta_1$ in  \eqref{eq:delta1}. Thus, we get that a particular solution by taking $C_{-}=0$ and $C_{+}$ as defined in \eqref{eq:C-plus-slow}
This determines $v^{\lam,1}$ up to an addition of a multiple of $v^{\lam,0}$ as in \eqref{eq:v1-slow} 
for any $\xi\in\R$.

\bibliographystyle{plainnat}
\small{\bibliography{references}}

\end{document}